\documentclass[11pt]{article}
\usepackage{fullpage}
\usepackage{amsmath} % the AMS-LaTeX package for mathematical formulae
\usepackage{amssymb} % the amsmath symbols
\usepackage{amsthm}  % the amsmath theorems and proofs package
\usepackage{verbatim}
\usepackage{color}

\begin{document}

%\pagestyle{plain}
% deafult for pagestyle is "plain" instead of "headings". "empty" can also be useful

%\renewcommand{\baselinestretch}{0.95}

\newcommand{\authornote}[2]{{\color{red} \bf { [[ #1:} #2 {\bf ]]}}}
\newcommand{\hnote}[1]{\authornote{Hai}{#1}}
\newcommand{\knote}[1]{\authornote{Kobbi}{#1}}

\newcommand{\remove}[1]{}

% switch the comment mark '%' on the next lines to control the size of the small matrices in the document
\newenvironment{smallbmatrix}{\left[\begin{smallmatrix}}{\end{smallmatrix}\right]}   % for small size
\newenvironment{smallpmatrix}{\left(\begin{smallmatrix}}{\end{smallmatrix}\right)}   % for small size

\def\D{{\cal D}}
\def\M{{\cal M}}
\def\R{{\cal R}}

%\def\RR{\hbox{\rm\setbox1=\hbox{I}\copy1\kern-.45\wd1 R}}
%\def\NN{\hbox{\rm\setbox1=\hbox{I}\copy1\kern-.45\wd1 N}}

% mathematical structures
\theoremstyle{definition} 		\newtheorem{mydefinition}{Definition}[section]
\theoremstyle{plain}	 				\newtheorem{mylemma}[mydefinition]{Lemma}
\theoremstyle{plain}					\newtheorem{myaxiom}[mydefinition]{Axiom}
\theoremstyle{plain}					\newtheorem{mytheorem}[mydefinition]{Theorem}
\theoremstyle{plain}					\newtheorem{myproposition}[mydefinition]{Proposition}
\theoremstyle{plain}					\newtheorem{myclaim}[mydefinition]{Claim}
\theoremstyle{plain}					\newtheorem{myobservation}[mydefinition]{Observation}
\theoremstyle{remark}					\newtheorem{mynote}[mydefinition]{Note}
\theoremstyle{remark}					\newtheorem{myexample}[mydefinition]{Example}

\begin{titlepage}

\author{Hai Brenner\thanks{Dept. of Mathematics,
Ben-Gurion University. {\tt haib@bgu.ac.il}.} 
\and Kobbi Nissim\thanks{Microsoft AI, Israel, and Dept. of Computer Science,
Ben-Gurion University. {\tt kobbi@cs.bgu.ac.il}.} \\
}

\title{Impossibility of Differentially Private Universally Optimal Mechanisms\thanks{Research partly supported by the Israel Science Foundation (grant No.\ 860/06), and by the Frankel Center for Computer Science, Dept.\ of Computer Science, Ben-Gurion University. }}

\date{\today}

% generates the title
\maketitle

\begin{abstract}
The notion of {\em a universally utility-maximizing privacy mechanism} was recently introduced by Ghosh, Roughgarden, and Sundararajan~[STOC 2009]. These are mechanisms that guarantee optimal utility to a large class of information consumers, {\em simultaneously}, while preserving {\em Differential Privacy} [Dwork, McSherry, Nissim, and Smith, TCC 2006]. Ghosh et al.\ have demonstrated, quite surprisingly, a case where such a universally-optimal differentially-private mechanisms exists, when the information consumers are Bayesian. This result was recently extended by Gupte and Sundararajan~[PODS 2010] to risk-averse consumers.

Both positive results deal with mechanisms (approximately) computing a {\em single count query} (i.e.,  the number of individuals satisfying a specific property in a given population), and the starting point of our work is a trial at extending these results to similar settings, such as sum queries with non-binary individual values, histograms, and two (or more) count queries. We show, however, that universally-optimal mechanisms do not exist for all these queries, both for Bayesian and risk-averse consumers.

For the Bayesian case, we go further, and give a characterization of those functions that admit universally-optimal mechanisms, showing that a universally-optimal mechanism exists, essentially, only for a (single) count query. At the heart of our proof is a representation of a query function $f$ by its {\em privacy constraint graph} $G_f$ whose edges correspond to values resulting by applying $f$ to neighboring databases.
\end{abstract}

\thispagestyle{empty}

\end{titlepage}

% inserts the table of contents
\thispagestyle{empty} \tableofcontents\newpage\setcounter{page}{1}

\section{Introduction}
{\em Differential Privacy}~\cite{DMNS06} is a rigorous notion of privacy that allows learning global (`holistic') information about a collection of individuals while preserving each individual's information private. The literature of differential privacy is now rich in techniques for constructing differentially privacy mechanisms, including some generic techniques such as the addition of Laplace noise with magnitude calibrated to global sensitivity~\cite{DMNS06}, addition of instance based noise calibrated to smooth sensitivity~\cite{NRS07}, and the exponential mechanism~\cite{MT07}. These and other techniques allow performing a wide scope of analyses in a differentially private manner, including conducting surveys over sensitive information, computing statistics, datamining, and sanitization. The reader is referred to~\cite{Dwork09} for a recent survey.

An immediate consequence of differential privacy is that (unless computing a constant function) a mechanism cannot compute a deterministic function. In other words, a differentially private version of an analysis would be a randomized approximation to the analysis, and furthermore, it would generally be possible to choose from a host of implementations for a task (e.g., the three generic techniques mentioned about may result with different mechanisms). Naturally, the designer of the analysis should choose one that is {\em useful}. Usefulness, however, depends on how the outcome of the analysis would be used, i.e., on the preferences of its {\em consumer}, that we henceforth refer to as an {\em information consumer}. 
Such a trade-off between uncertainty and utility, while taking consumer's preferences into account, is the subject of rational-choice theory and decision theory, as noted in \cite{GRS09, GS10}.

We discuss the two models of utility which were previously discussed in \cite{GRS09, GS10}. In both, the information consumer has \emph{side information} (her own world-view or previous knowledge), and a \emph{loss-function} which quantifies the consumer's preferences and the quality of the solution for her problem. Intuitively, it describes how bad is a deviation from the exact answer for the consumer, a measure of her intolerance towards the inaccuracy  imposed by differentially private mechanisms. Finally, the models assume that the consumers are \emph{rational} - they combine the structure of the mechanism, their side information and their personal loss-function (preferences) with the goal of minimizing their loss, or, equivalently maximizing their utility. The two models differ in the way side information is formulated and respectively how utility function is defined. 
Subject to the requirements of differential privacy, one usually has a choice from a collection of implementations. As discussed in decision-theory and assuming rational information consumers, each consumer will choose a mechanism which maximizes her utility. This is an \emph{optimal} mechanism for this consumer.

Information consumers' accuracy requirements vary: for some consumers only an exact answer would be of value, whereas others may aim at minimizing the estimate bias ($\ell_1$ error), or its variance ($\ell_2$ error), and, clearly, many other criteria exist. It seems that a discussion of the utility of differentially private mechanisms should take this rich variety into account. The recent work of Ghosh, Roughgarden, and Sundararajan~\cite{GRS09} has put forward a serious attempt at doing exactly that with respect to (oblivious) Bayesian information consumers. In this utility model, the consumer's side information is described as an a priori distribution on the exact result of the analysis. The recent work of Gupte and Sundararajan~\cite{GS10}  considers a related model where the information consumers are {\em risk-averse}. Here, the information consumer's knowledge is a set of possible values the exact analysis can take, and an optimal mechanism minimizes the consumer's worst-case expected loss.

Composition theorems for differential privacy only guarantee that the degradation in privacy is not more than exponential in the number invocations. Hence, while different consumers may exhibit different optimal mechanisms, a very important goal is to avoid invoking that  multiplicity of mechanisms. This degradation is part of the motivation for the work on {\em sanitization} where a family of queries are answered at once~\cite{DwNi04,BLR08,FFKN09,DMRRV09}, the work on {\em privacy under continual observation}~\cite{DNPR10}, and the construction of the {\em Median Mechanism}~\cite{RR10}. A surprising result of Ghosh, Roughdarden, and Sundararajan~\cite{GRS09} is that invoking a multiplicity of optimal mechanisms may not be necessary. They consider a database that is a collection of Binary inputs (e.g., pertaining to having some disease) and Bayesian information consumers that wish to count the number of {\em one} entries in the database (equivalently, compute the sum of the entries). They show the existence of a single mechanism that enables optimality for {\em all} Bayesian information consumers (the mechanism needs to be invoked only once). The mechanism itself is not optimal for all Bayesian information consumers, however, each consumer can perform a deterministic remapping on the outcome of the common mechanism, where the remapping is chosen according to her notion of utility, and locally output a result that is effectively according to one of her optimal mechanism.  Such a common mechanism is referred to as {\em universally optimal}. An analogous result for risk-averse information consumers was shown in~\cite{GS10}. 

Are these results of~\cite{GRS09} and~\cite{GS10} that deal with the simple case of a single count query ``accidental'', or can they be extended to other queries? to multiple queries? One would anticipate that universally-optimal mechanisms should exist (at least) for those queries that are closely related to counting, such as sum queries where the inputs are non-binary, histograms, and bundles of two or more count queries.

\subsection{Our Results and Directions for Future Progress}

In contrast with the anticipation expressed in the previous paragraph, we show that settings in which universally optimal mechanisms exist are extremely rare, and, in particular, in both the setting of Bayesian and of risk-averse information consumers, universally optimal mechanisms do not exist even for sum queries where the inputs are non-binary, histograms, and bundles of two or more count queries. 

Moreover, in the case of Bayesian information consumers, we give a characterization of those functions of the data that admit universally optimal mechanisms. The characterization makes use of a combinatorial structure of the query function $f:\D^n\rightarrow \R_f$, where $\D$ is the domain of the database records and $\R_f$ is the output space of the query function. We define this combinatorial structure of the query $G_f$ and call it a {\em privacy constraint graph}. The vertices of $G_f$ correspond to values in $\R_f$, and edges correspond to pairs of values resulting by applying $f$ to neighboring databases. (This graph was examined in some proofs in \cite{KL10} as well). We show:

\medskip\noindent{\bf Theorem~\ref{thm:acyclic} {\rm (Informal)}.} {\em If $G_f$ contains a cycle then no universally optimal mechanism exists for $f$. }

\medskip\noindent{\bf Theorem~\ref{thm:delta3} {\rm (Informal)}.} {\em If $G_f$ is a tree that contains a vertex of degree 3 or more, then no universally optimal mechanism exists for $f$ for better values of the privacy parameter.}

\medskip

Facing the impossibility of universal optimality, an alternative may be found in an approximate notion, which enables (approximate) optimality to (approximately) all of the information  consumers. A good notion of approximate optimality should allow constructing such mechanisms for sum queries, histograms, and more. Furthermore, it should allow performing several queries and satisfy a composition requirement, in a sense that when applying two such mechanisms to two different queries, the resulting composed mechanism should be somewhat approximately optimal for the two queries together. 

Finally, we note that, following prior work we focus on {\em oblivious} mechanisms (see Section~\ref{sec:obliviousMehcanisms} for the technical definition). In Section \ref{sec:count_generalizations}, we show that for the intuitive generalizations of count queries, enabling \emph{non}-oblivious universal mechanisms from which optimal oblivious mechanisms are derived, still leaves the construction of universally optimal mechanisms impossible. The question whether non-oblivious universally-optimal mechanisms exist for some other natural abstract queries, from which all oblivious universally-optimal mechanisms may be derived is left open.

\subsection{Related Work}

Most relevant to our work are the papers by Ghosh, Roughgarden, and Sundararajan~\cite{GRS09} and  by Gupte and Sundararajan~\cite{GS10}. Ghosh et al.\ show that the geometric mechanism (a discrete version of the Laplace mechanism of~\cite{DMNS06}) yields optimal utility for all Bayesian information consumers for a count query. Their proof begins by observing that all differentially private mechanisms correspond to the feasible region of a Linear Program (a polytope), and that minimizing disutility can be expressed as minimizing a linear functional. Hence, every Bayesian information consumer has an optimal mechanism corresponding to a vertex of the polytope, which in turn corresponds to a subset of the constraints of the Linear Program which are tight (optimal mechanisms, not corresponding to the polytope vertices, may also exist). 
They introduce a {\em constraint matrix} that uniquely corresponds to a vertex of the polytope, and indicates which constraints are tight, and which are slack on that vertex. Those constraint matrices that correspond to optimal mechanisms, are shown to have some special structure that allows to derive mechanisms with the same signature (and thus equal) from the geometric mechanism using some deterministic remapping on its output. 

We are also interested in observing the tight constraints in some mechanisms. We will not need the full description of the structure of such a constraint matrix. Instead we only use the observation that tight privacy constraints can be derived only from mechanisms that also obey similar tight constraints.

Gupte and Sundararajan show similar results for the risk-averse utility model, where consumers try to minimizes their maximal worst-case disutility. They provide a full characterization of the mechanisms which are derivable (by random remapping) from the geometric mechanism and use this characterization to construct a universally-optimal mechanism for a count query. An interesting feature of the construction is that it releases noisy answers of the query at different privacy levels, thus keeping more privacy against specific consumers, and enabling more utility to others. 

Also related to our work is the recent work of Kifer and Lin~\cite{KL10} that studies privacy and utility, in a very general setting, from an axiomatic point of view. 
%Kifer et al.~also examine the \emph{privacy constraint graph} in some proofs, but quite differently: they are interested in abstract privacy and utility models whereas we  focus on differential privacy and utility models described in \cite{GRS09, GS10}. 
They introduce a partial order on mechanism where mechanism $Y$ is at least as \emph{general} as mechanism $X$ if $X$ can be derived from $Y$ by post processing. They also introduce the concept of maximal generality, which turns to be useful in our proofs.

\section{Preliminaries}

\subsection{Differential Privacy~\cite{DMNS06}}

Simply speaking, a mechanism which preserves differential-privacy will output for any two databases which `look alike' the same result, with similar probabilities. 
More formally, consider databases  $D_1,D_2 \in \D^n$ which consist of $n$ records out of some domain $\D$. The Hamming Distance between $D_1$ and $D_2$ is the number of records on which they differ. We will call databases at distance one {\em neighboring}.
\begin{mydefinition}[Differential Privacy~\cite{DMNS06}]
\label{def:differential}
Let $\M:\D^n \rightarrow \R$ be a probabilistic mechanism.
$\M$ preserves \emph{$\alpha$-differential-privacy} for $\alpha \in \left(0, 1\right)$ if for any two neighboring databases $D_1,D_2 \in \D^n$ and any (measurable) subset of the mechanism's range $S\subseteq \R$,
\begin{equation}
	Pr \left[ \M(D_1) \in S \right] \ge \alpha \cdot Pr \left[ \M(D_2) \in S \right].
\end{equation}
The probability is taken over the coin tosses of the mechanism $\M$.
\end{mydefinition}
Notice that the greater $\alpha$ is the less the mechanism's output depends on the exact query result, and so better privacy is attained. %(We exclude the case $\alpha=1$ as it excludes learning any information  about the input database.)

\subsection{Oblivious Mechanisms}
\label{sec:obliviousMehcanisms}
We consider a setting where several information consumers are interested in estimating the value of some query $f(\cdot)$ applied to a database $D\in\D^n$, and answered by a differentially private mechanism $\M$. 
Ghosh et al.~\cite{GRS09} show that if no restriction is put on the mechanism, then no universally optimal mechanism exists for count queries (intuitively, universal optimality, defined below, means that all potential consumers minimize their loss simultaneously). On the other hand, universally optimal mechanisms sometimes do exist if we restrict our mechanisms such that their output distribution depends only on the the exact query result (a.k.a.\ {\em oblivious mechanisms}). This is why in~\cite{GRS09} (and later in~\cite{GS10}) only oblivious mechanisms are considered\footnote{Impossibility of universal optimality when the mechanisms are not restricted to being oblivious is proved in~\cite{GRS09} for Bayesian information consumers. For risk-averse consumers,~\cite{GS10} show that non-oblivious mechanisms may be replaced with oblivious ones without affecting the consumers' utility for the worse.}. 
We follow suit and only consider oblivious mechanisms.
We show in Subsection \ref{sec:m_gt_2} that this restriction does not weaken the basic results presented in Section \ref{sec:count_generalizations}.

\begin{mydefinition}[Oblivious Mechanism]
\label{oblivious}
Let $f : \D^n \rightarrow \R_f$ be a query. A mechanism $\mathcal{M}:\D^n\rightarrow \R$ is $f$-\emph{oblivious} (or simply {\em oblivious}) if there exists a randomized function $\tilde\M:\R_f\rightarrow\R$ such that, for all $D\in \D^n$, the distributions induced by $\M(D)$ and $\tilde\M(f(D))$ are identical.
\end{mydefinition}

Combining $\alpha$-differential privacy with obliviousness, we get that for every $i,i' \in \R_f$ which are outputs of neighboring databases $D,D'$ (i.e., $f(D) = i$ and $f(D')=i'$), then 
$\Pr[\tilde\M(i)\in S] \geq \alpha \cdot \Pr[\tilde\M(i')\in S]$ for all $S\subseteq \R$.

\subsubsection{Oblivious Differentially Private Mechanisms for a Count Query}

An oblivious finite-range mechanism $\M:\D^n\rightarrow \R$ estimating $f:\D^n\rightarrow\R_f$ can be described by a row-stochastic matrix $X=\left(x_{i,j}\right)$ of the underlying randomized mapping $\tilde\M$, whose rows are indexed by elements of $\R_f$, and whose columns are indexed by elements of $\R$, where $x_{i,j}$ equals the probability of outputting $j\in\R$ when $f(D) = i$.
Since $\R$ is finite, and information consumers anyway remap the outcome of $\M$, we can assume, wlog, that $\R=\{0,1,2,\ldots,|\R|-1\}$. 

We now consider the case where $\D=\{0,1\}$ and $f(D)$ counts the number of one entries in $D$. Hence, $\R_f=\{0,\ldots,n\}$ and the matrix $X$ is of dimensions $(n+1)\times |\R|$. Preserving $\alpha$-differential privacy poses constraints on the transition matrix $X$ beyond row-stochasticity. Note that for the count query, the query results of two neighboring databases may differ by at most one.  Differential privacy hence imposes the constrains $x_{i,j} \ge \alpha \cdot x_{i+1,j}$ and $x_{i+1,j} \ge \alpha \cdot x_{i,j}$ where $i \in \R_f=\{ 0 \ldots n-1 \}$ and $j\in\R$. Adding row-stochasticity and differential privacy, we get that an oblivious differentially private mechanism for the count query should satisfy the following linear constraints:
\begin{align}
		x_{i,r} \ge \alpha x_{i+1,r} \qquad & \forall i \in \{0,\ldots,n-1\}, \forall r \in \R  \label{eq:LP_first} \\
		\alpha x_{i,r} \le x_{i+1,r} \qquad & \forall i \in \{0,\ldots,n-1\}, \forall r \in \R  \label{eq:LP_second}\\
		\sum_{r\in\R}  x_{i,r} = 1 \qquad & \forall i \in \{0,\ldots,n\}  \\
		x_{i,r}\ge 0  \qquad & \forall i \in \{0,\ldots,n\}, \forall r \in \R \label{eq:LP_last}
\end{align}

\subsection{Utility Models}

We use the utility models defined in~\cite{GRS09} and~\cite{GS10}. In both, a \emph{loss function}  $\ell(i,r)$ quantifies an information consumer's disutility when she chooses to use answer $r$ while the correct answer is $i$. Given a loss function $\ell(\cdot,\cdot)$ of an information consumer, if the exact answer is $i$ then her expected loss is \mbox{$\sum_{r\in\R} x_{i,r} \cdot \ell(i,r)$}.\footnote{This is only true if the consumer uses the mechanism $X$ \emph{directly}, i.e., the consumer leaves the mechanism's output as is, and does not apply a post-processing step. The ability to apply such a post-processing step on the mechanism's output will be discussed in the next sub-section.} Loss functions vary between consumers, and the only assumptions made in~\cite{GRS09,GS10} is that $\ell(i,r)$ depends on $i$ and $|i-r|$ and is monotonically non-decreasing in $|i-r|$ for all $i$. (This is a reasonable requirement that turns to be crucial for the existence of a universally optimal mechanism~\cite{GRS09}.) 
Examples of loss functions include  $\ell_1(i,r) = |i-r|$ (consumers who care to minimize expected mean error); $\ell_2(i,r) = (i-r)^2$ (minimize error variance); and  $\ell_{bin}(i,r)$ that  evaluates to $0$ if $i=r$ and  to $1$ otherwise (minimize number of errors). 

Information consumers differ in their knowledge about the exact $f(D)$. References \cite{GRS09} and~\cite{GS10} model this knowledge differently as we now describe.

\paragraph{Bayesian Model~\cite{GRS09}}

In the Bayesian utility model, an information consumer's knowledge is represented by a vector $\bar p$ where $p_i$ is the consumer's a~priori probability that $f(D)=i$. Having a vector of prior probabilities $\bar p$ and loss function $\ell(\cdot,\cdot)$,  the consumer's expected loss can be expressed as 
\mbox{$\sum_i p_i \cdot \sum_r x_{i,r} \cdot \ell(i,r)$}. The {\em optimal mechanisms}  for this information consumer hence are the solutions of the linear program in the variables $x_{i,r}$ consisting the constraints in Equations~(\ref{eq:LP_first})-(\ref{eq:LP_last}) and the objective
\begin{equation}
	\text{minimize}~\sum_{i\in \R_f} p_i \cdot \sum_{r\in\R}  x_{i,r} \cdot \ell(i,r).
\end{equation}

\paragraph{Risk-Averse Model~\cite{GS10}}

In the risk-averse utility model an information consumer's knowledge restricts the possible values for the exact $f(D)$. This is expressed by a set $S\subseteq\R_f$ of the possible values $f(D)$ can take. The consumer is interested in minimizing her maximal expected loss conditioned on $f(D)\in S$, i.e., $\max_{i \in S} \sum_r x_{i,r}\cdot \ell(i,r)$.
Similarly to the above, the optimal mechanism for an information consumer is a solution to a linear program consisting the constraints in Equations~(\ref{eq:LP_first})-(\ref{eq:LP_last}) and the objective
\begin{equation}
	\text{minimize  } \max_{i \in S} \sum_{r\in \R}  x_{i,r}\cdot \ell(i,r).
\end{equation}

\subsection{Remapping and Generality}

An information consumer might have access to a private mechanism $U$ which is not tailored specifically for her needs (i.e., to her prior knowledge and loss function). Yet, she may be able to recover a better mechanism for her needs by means of post-processing, which we will denote \emph{remapping}. To intuit remapping, consider a consumer that knows that for the specific database the count query cannot yield the answer $0$. If that consumer receives a $0$, it may be beneficial for her to remap it to $1$. (Recall that the loss function is monotone in $|i-r|$.) Denoting the given mechanism by $U$ and the remapping by $T$ (a row-stochastic linear transformation, $T$ has no access to the information of the database other then the output of $U$), the actual mechanism that is used by the information consumer is denoted $T \circ U$ (in matrix form: $UT$).

Notice that given a mechanism $U$ with a finite range, an information consumer can find the optimal remapping $T$ for her (such that $T \circ U$ has optimal utility), by constructing a linear program in which $T = (t_{i,j})$ are the program variables \cite{GS10}.

\begin{mydefinition}[Derivable Mechanisms, Generality Partial Order \cite{KL10}] \label{def:derivable}
Let $X,Y$ be private mechanisms. We say that a mechanism $X$ is \emph{derivable} from a mechanism $Y$ if there exists a random remapping $T$ of the results of mechanism $Y$, such that $X= T \circ Y$. We also say that $Y$ is \emph{at least as general as $X$}, and denote this relation by $X \preceq _G Y$.
If $X \preceq _G Y$ and $Y \preceq _G X$ we say that $X,Y$ are \emph{equivalent}. 
\end{mydefinition}

\begin{mydefinition}[Maximal Generality \cite{KL10}]
\label{def:maximallygeneral}
Let $X$ be an $\alpha$-differentially private mechanism. $X$ is \emph{maximally general} if for every $\alpha$-differentially private mechanism $Y$, if $X \preceq _G Y$ then $Y \preceq _G X$.
\end{mydefinition}

After introducing the notion of maximally general mechanisms (for any definition of privacy), Kifer et al.~fully characterize all maximally general private mechanisms with a finite input space in the differential privacy setting. First they introduce the concept of \emph{column-graphs}\footnote{Kifer et al.~actually define \emph{row graphs} and not \emph{column graphs}. We follow the matrix structure of \cite{GRS09, GS10} which is simply the transposed matrix of the one used by Kifer et al., hence the difference in terminology.} of a private mechanism, which mark the tight privacy constraints in one column of the mechanism $X$.

\begin{mydefinition}[Column graph \cite{KL10}]
\label{def:columngraphs}
Let $X$ be an $\alpha$-differentially private mechanism with a finite input space. Let $r$ be some possible output of $X$, and $x_r$ be its corresponding column in $X$. Let $I$ be the input space of $X$ (corresponding to $X$'s rows). The graph associated with this column has $I$ as the set of nodes, and for any $i_1,i_2 \in I$, there is a directed edge $(i_1, i_2)$ if $i_1$ and $i_2$ match neighboring databases and $x_{i_1,r} = \alpha x_{i_2,r}$, and a directed edge $(i_2, i_1)$ if $x_{i_2,r} = \alpha x_{i_1,r}$. The direction of the edges is only necessary to distinguish between maximally general mechanisms which have similar undirected column-graphs, but it will not be essential to the rest of this article.
\end{mydefinition}

Kifer and Lin characterize the maximally general differentially private mechanisms with a finite input space:

\begin{mytheorem}[\cite{KL10}]
\label{thm:maximal_characterization}
Fix a privacy parameter $\alpha$ and a database query $f$ with a finite range for databases of a specific size. Let $X$ be an $\alpha$-differentially private mechanism with a finite range. Then $X$ is maximally general iff each column graph of $X$'s columns (according to the privacy constraints implied by $f$) is connected.
\end{mytheorem}

This theorem shows that we wish to maximize the set of tight privacy constraints in order to make a private mechanism as general as possible. Notice that having just one entry of a column in $X$ and the spanning tree of this column's graph (we need to know the direction of the edges as well), determines all the entries of this column.

\subsection{Universal Mechanisms}
Consider a collection of Bayesian information consumers, and suppose we wish to enable each of the information consumers to sample a result from a differentially private mechanism optimizing her utility. Ghosh et al.~\cite{GRS09} showed that this does not necessarily require executing multiple mechanisms: if the query is a count query, then it is possible to construct one {\em universally optimal} mechanism $U$, from which all information consumers can {\em simultaneously} recover an optimal mechanism for their needs by {\em remapping}. I.e., every information consumer has an optimal private mechanism which is derivable from $U$. This result is repeated for risk-averse information consumers by Gupte et al.~\cite{GS10}. More formally:

\begin{mytheorem}[Universal optimality, Bayesian consumers~\cite{GRS09}] \label{thm:UniversalOptimalityBayesian}
Fix a privacy parameter $\alpha\in(0,1)$. There exists an $\alpha$-differentially private mechanism $U$ for a single count query, such that for every prior $\bar p$ and every monotone loss function $\ell(\cdot,\cdot)$ there exists a (deterministic) remapping $T$ such that $T \circ U$ implements an optimal oblivious mechanism for $\bar p, \ell(\cdot,\cdot)$.
\end{mytheorem}

\begin{mytheorem}[Universal optimality, risk-averse consumers~\cite{GS10}]\label{thm:universalOptimalityRiskAverse}
Fix a privacy parameter $\alpha\in(0,1)$. There exists an $\alpha$-differentially private mechanism $U$ for a single count query, such that for every set $S$ of possible outcomes and every monotone loss function $\ell(\cdot,\cdot)$ there exists a (probabilistic) remapping $T$ such that $T \circ U$ implements an optimal oblivious mechanism for $S, \ell(\cdot,\cdot)$.
\end{mytheorem}

It turns out that in both theorems $U$ is realized by the geometric mechanism -- a variant of the mechanism adding Laplace noise of~\cite{DMNS06}.
Note that there may be optimal mechanisms which cannot be derived from the geometric mechanism, but for every information consumer there is at least one private mechanism that is derivable from the geometric mechanism and is optimal for her.

\section{Impossibility of Universally Optimal Mechanisms for Generalizations of Count Queries}
\label{sec:count_generalizations}

When the domain of the database records is $\{0,1\}$, a count query is equivalent to a sum query. Theorems~\ref{thm:UniversalOptimalityBayesian}  and~\ref{thm:universalOptimalityRiskAverse} can hence be thought of as applying to a {\em sum query} over the integers, where the domain of the database is Binary. It is natural to ask whether the results of these theorems can be extended to showing that universally optimal mechanisms exist for sum queries when the underlying data is taken from a larger domain such as \mbox{$\D = \{0,1,\ldots,m\}$} where $m\ge 2$. We answer this question negatively.

Consider the case $m=2$. Recall that an oblivious differentially private mechanism can be described by a row-stochastic matrix $X = (x_{i,j})$, such that $x_{i,j}$ is the probability of the mechanism to return $j$ when the exact result is $i$.
A difference of the case $m=2$ from count queries ($m=1$) is that applying a sum query to two neighboring databases may yield results which differ by 0, 1, or 2 (instead of 0 or 1). Therefore, in the linear program describing mechanism $X$ equations~(\ref{eq:LP_first}) and (\ref{eq:LP_second}), should be replaced by the following four constraints (the range for $i$ in the other equations should be modified to $0,\ldots,2n$):
\remove{\begin{align*}
		x_{i,r} \ge \alpha x_{i+1,r} \qquad & \forall i \in \{0,\ldots,2n-1\}, \forall r \in \R  \\
		\alpha x_{i,r} \le x_{i+1,r} \qquad & \forall i \in \{0,\ldots,2n-1\}, \forall r \in \R \\
		x_{i,r} \ge \alpha x_{i+2,r} \qquad & \forall i \in \{0,\ldots,2n-2\}, \forall r \in \R  \\
		\alpha x_{i,r} \le x_{i+2,r} \qquad & \forall i \in \{0,\ldots,2n-2\}, \forall r \in \R 
\end{align*}}
\begin{align*}
& x_{i,r} \ge \alpha x_{i+1,r}, \quad \alpha x_{i,r} \le x_{i+1,r} \quad & \forall i \in \{0,\ldots,2n-1\}, \forall r \in \R  \\
& x_{i,r} \ge \alpha x_{i+2,r}, \quad \alpha x_{i,r} \le x_{i+2,r} & \forall i \in \{0,\ldots,2n-2\}, \forall r \in \R
\end{align*}
Once again, a consumer's optimal mechanism can be found by solving a linear program with all the constraints and the appropriate target function.

\subsection{The Basic Impossibility Result for Sum Queries}
\label{sec:sum_queries}

We first consider the case where the database contains $n=1$ record, taking values in \{0,1,2\} (i.e., $m=2$). 
Later, we generalize to $n\geq 1$ and $m\geq 2$.
Note that in the case of $n=1$, the non-oblivious mechanisms are identical to oblivious mechanisms. We consider non-oblivious universal mechanisms as well when generalizing this result lo larger values of $n$.
\begin{myobservation}
\label{obsrv:Bayesian}
In the Bayesian model there exists an information consumer whose only optimal mechanism is $X = \frac{1}{1+2\alpha} \cdot 
\begin{smallbmatrix}
	1 & \alpha & \alpha  \\
	\alpha & 1 & \alpha  \\
	\alpha & \alpha & 1
\end{smallbmatrix}$ and an information consumer whose optimal mechanisms are all of the form 
$ Y = \frac{1}{1+\alpha} \cdot 
\begin{smallbmatrix}
	1 & \alpha & 0 \\
	\alpha & 1 & 0 \\
	q & 1+\alpha-q & 0       
\end{smallbmatrix}$, 
where $q \in \left[\alpha, 1\right]$.
\end{myobservation}

\begin{proof}
Consider an information consumer with a prior $\bar p = (\frac{1}{3},\frac{1}{3}, \frac{1}{3})$ and a loss function $\ell_{bin}$ (i.e., a penalty of $1$ whenever she chooses an answer different from the exact result, and no penalty otherwise). It is easy to see that no optimal mechanism for this consumer outputs a value not in $\{0,1,2\}$. 

The information consumer wishes to minimize 
$$\sum_{i=0}^2 p_i  \sum_{r=0}^2 x_{i,r} \cdot \ell(i,r) 
=  \frac{1}{3} \sum_{i=0}^{2} \sum_{r\ne i} x_{i,r} 
= \frac{1}{3} \sum_{i=0}^{2} \left( 1-x_{i,i} \right)  = 1 - \frac{1}{3} \sum_{i=0}^{2} x_{i,i}.$$
And so, the consumer's goal is to maximize $\sum_{i=0}^{2} x_{i,i}$ subject to maintaining $\alpha$-differential privacy.

For $i \in \{0,1,2\}$, having $\alpha$-differential privacy implies 
\begin{equation} \label{eq:DPineq}
\alpha x_{i,i} \le  x_{j,i} \quad \forall j \in \{0,1,2\}\setminus \{i\},
\end{equation}
and hence (by summing up Equation~(\ref{eq:DPineq}) for $j\not=i$), we get
\begin{equation} \label{eq:sumeq}
2 \alpha x_{i,i} 
	=   \sum_{\substack{j=0 \\ j \ne i}}^{2} \alpha x_{i,i} 
	\le \sum_{\substack{j=0 \\ j \ne i}}^{2} x_{j,i}. 
\end{equation}
Summing up Equation~(\ref{eq:sumeq}) for $i \in \{0,1,2\}$ we get 
\begin{equation*}
	\sum_{i=0}^{2} 2 \alpha x_{i,i} 
	\le \sum_{i=0}^{2} \sum_{\substack{j=0 \\ j \ne i}}^{2} x_{j,i} 
	= \sum_{i=0}^{2} (1-x_{i,i})
	= 3 - \sum_{i=0}^{2} x_{i,i},
\end{equation*}
and we can now conclude that $\sum_{i=0}^{2} x_{i,i} \le \frac{3}{2 \alpha +1}$.
This inequality is tight iff Equation~(\ref{eq:DPineq}) is tight (i.e., $x_{j,i} = \alpha x_{i,i}$) for every $i\not=j$. In that case, we get the following system of linear equations: 
\begin{align*}
	x_{11} + \alpha x_{22} + \alpha  x_{33}  & =  1\\
	\alpha x_{11} + x_{22} + \alpha  x_{33}  & =  1\\
	\alpha x_{11} + \alpha x_{22} +  x_{33}  & =  1
\end{align*}
Since the three equations are linearly independent, we get a {\em unique} solution: $x_{1,1}=x_{2,2}=x_{3,3}=\frac{1}{1+2\alpha}$. 

A similar proof shows that mechanisms of the form $Y$ are the only mechanisms optimal for information consumers with a prior $p_0 = p_1 = \frac{1}{2}, p_2=0$ and loss function $\ell_{bin}$. 

It may seem like we restrict ourselves only to information consumers with the $\ell_{bin}$ loss function. Note that, according to Theorem \ref{thm:maximal_characterization}, there are not so many maximally general mechanisms whose range is a subset of $\{ 0, 1, 2 \}$, and some of them are not optimal for any consumer. 
Therefore, the mechanisms described are also the only optimal mechanisms for a variety of other information consumers, such as whose prior is $p_0 = p_1 = \frac{1}{2},\, p_2=0$ and loss function is $\ell_1$. Also, even more such consumers can be found easily in any sequence of consumers which converge to consumers with such unique optimal mechanisms (i.e., their priors and loss functions converge to the prior and loss function of the consumer we chose). Such information consumers with close priors and close loss functions to the ones described above will have the same unique optimal mechanisms.
\end{proof}

\begin{myobservation}
\label{obsrv:risk-averse}
In the risk-averse model there exists an information consumer whose only optimal mechanism is  $X = \frac{1}{1+2\alpha} \cdot 
\begin{smallbmatrix}
	1 & \alpha & \alpha  \\
	\alpha & 1 & \alpha  \\
	\alpha & \alpha & 1
\end{smallbmatrix}$ and an information consumer whose optimal mechanisms are all of the form 
$ Y = \frac{1}{1+\alpha} \cdot 
\begin{smallbmatrix}
	1 & \alpha & 0 \\
	\alpha & 1 & 0 \\
	q & 1+\alpha-q & 0       
\end{smallbmatrix}$, 
where $q \in \left[\alpha, 1\right]$.
\end{myobservation}
\begin{proof}
Consider an information consumer whose loss function is $\ell_{bin}$ who knows the support of the query is $S=\{0,1,2\}$. As in the previous observation, the support of any optimal mechanism for this consumer must be a subset of $\{0,1,2\}$. Notice that if the consumer uses the mechanism described by $X$ then her maximal expected loss is $\frac{2\alpha}{1+2\alpha}$. 

Assume for a contradiction that the consumer has another mechanism $X'$ with maximal expected loss at most $\frac{2\alpha}{1+2\alpha}$. I.e., 
\begin{equation} \label{eq:lossAltMech}
	\max \{x'_{0,1}+x'_{0,2}, x'_{1,0}+x'_{1,2}, x'_{2,0}+x'_{2,1}\} \le \frac{2\alpha}{(1+2\alpha)}.
\end{equation} 
Since $X' \not= X$, Equation~(\ref{eq:lossAltMech}) implies that $x'_{i,j}<\frac{\alpha}{1+2\alpha}$ for some
 $i\ne j$. Taking into account that $X'$ is $\alpha$-differentially private we get 
 $x'_{j,j} \leq \frac{1}{\alpha} \cdot x'_{i,j} < \frac{1}{1+2\alpha}$, 
 and hence the maximal expected loss is at least $\sum_{i\not=j} x'_{i,j} = 1-x'_{j,j} > 1-\frac{1}{1+2\alpha} = \frac{2\alpha}{1+2\alpha}$, in contradiction to the assumption that this mechanism is at least as good as $X$ for this information consumer.

A similar proof shows that mechanisms of the form $Y$ are the only mechanisms optimal for an information consumer with auxiliary knowledge of the support $S=\{0,1\}$ and loss function $\ell_{bin}$. As in the previous observation, the mechanisms described are also the only optimal mechanisms for a variety of other information consumers.
\end{proof}

We will now use these two observations to show that in both models no universally optimal mechanism $U$ exists. (This is true even if we allow $U$ to have a non-discrete range.)

\begin{myclaim}
\label{clm:impossibility_sumquery}
No $\alpha$-differentially private mechanism can derive both $X$ and an instance of $Y$.
\end{myclaim}
\begin{proof}
Assume for a contradiction that such a mechanism $U$ exists, so $X$ and some instance of $Y$ are both derivable from $U$. For simplicity we refer to this instance as $Y$. By Theorem \ref{thm:maximal_characterization}, $X$ is a maximally general mechanism. Therefore $U \preceq _G X$, and hence $Y \preceq _G X$, i.e.,  there exists a random remapping $T$ such that $Y = XT$.
Denote by $x_j$ the $j^\text{th}$ column of $X$, and by $y_k$ the $k^\text{th}$ column of $Y$. We get that
\begin{equation*}
	y_k = t_{0,k} \cdot x_0 + t_{1,k} \cdot x_1 + t_{2,k} \cdot x_2 , \qquad \forall k \in \{0,1,2\}
\end{equation*}
Note that some $\alpha$-differentially privacy constraints in $Y$ are tight. Specifically, $y_{1,0} = \alpha y_{0,0}$ and $y_{0,1} = \alpha y_{1,1}$. As $Y$'s columns are non-negative linear combinations of $X$'s columns, such a tight constraint in a column of $Y$ appears only if this column is a linear combination of columns of $X$ in which the same privacy constraints are also tight. Note that the first two entries of every column in $Y$ correspond to a tight constraint. But since $x_{0,2}=x_{1,2}>0$, mapping this column of $X$ by $T$ to any column of $Y$ (even with just a positive probability), yields a mechanism with a column in which the first two entries do not correspond to a tight constraint. Therefore, a contradiction.
\end{proof}

\subsection{Generalizing the Impossibility Result for Sum Queries}
\label{sec:sum_generalizations}

So far we have shown the following: if $n$, the number of records in the database, is $1$, and the range of values is $0,\ldots,m$ where $m=2$, then no universal private mechanism for sum queries yields optimal utility for all consumers. Next, we generalize these impossibility results to the case $m \geq 2$ (and $n=1$), and later present also the case where $n>1$. Hence, we will conclude the following theorem:
\begin{mytheorem}
\label{thm:sum_queries}
No universally optimal mechanism exists for sum queries for databases whose records take values in the set $\{0,1,\ldots ,m\}$ where $ m \geq 2$. This holds both for the Bayesian and the risk-averse utility models.
\end{mytheorem}

\subsubsection{Generalizing the Sum Query Impossibility Result to $m > 2$}
\label{sec:m_gt_2}

Consider the case where the database consists of one record, and the possible values in this record are $0$ to $m$. Let 
\begin{equation} \label{eq:mechanismsXY}
       X = \frac{1}{1+m \alpha} \cdot 
	\begin{bmatrix}
		1 & \alpha & \alpha & \cdots & \alpha \\
		\alpha & 1 &\alpha  & \cdots & \alpha \\
		\alpha & \alpha & 1 & \cdots & \alpha \\
		\vdots & \vdots & & \ddots & \vdots \\
		\alpha & \alpha &\alpha &\cdots &1
	\end{bmatrix};\;
	Y = \frac{1}{1+(m-1) \alpha} \cdot 
	\begin{bmatrix}
		1 & \alpha & \alpha & \cdots & \alpha & 0\\
		\alpha & 1 &\alpha  & \cdots & \alpha & 0\\
		\alpha & \alpha & 1 & \cdots & \alpha & 0\\
		\vdots & \vdots & & \ddots & \vdots \\
		\alpha & \alpha &\alpha &\cdots & 1 & 0\\
		q_1 & q_2 & q_3 & \cdots & q_{m} & 0
	\end{bmatrix},
\end{equation}
where $\alpha \leq q_i \leq 1$ and $\sum_{i=1}^m q_i =1+(m-1)\alpha$.
Similar arguments to those used for the case $m=2$ show that $X$ is the unique optimal mechanism for an information consumer with loss function $\ell_{bin}$ and prior $p_0=p_1=\cdots = p_m = \frac{1}{m+1}$ in the Bayesian utility model and for an information consumer with support $S=\{0,1,\ldots ,m\}$ in the risk-averse utility model. Also, mechanisms of the form $Y$ are the only optimal mechanisms for the information consumers with loss function $\ell_{bin}$ and prior $p_0=p_1=\cdots = p_{m-1} = \frac{1}{m},\, p_m=0$ in the Bayesian model, and for an information consumer with support $S=\{0,1,\ldots ,m-1\}$ in the risk-averse model. Once again, these mechanisms are also the only optimal private mechanisms for a variety of other consumers as well. Using the same arguments to those in the proof of Claim \ref{clm:impossibility_sumquery}, it follows that $X$ and $Y$ are not derivable from one single mechanism.

\subsubsection{Generalizing the Sum Query Impossibility Result to $n>1$}
\label{sec:n_gt_1}

Now consider the case where the number of records in the database is larger than $1$. We first prove the impossibility of an {\em oblivious} universally optimal mechanism. Consider two consumers with loss function $\ell_{bin}$. The first consumer believes that the result of the sum query is bounded by $m$ (in the Bayesian case, the consumer holds a uniform prior over $\{0,\ldots,m\}$). No optimal mechanism for this consumer returns values larger than $m$, so in the mechanism matrix the columns corresponding to values greater than $m$ contain zeros. Refer to some optimal mechanism for this consumer as $X'$. Ignoring rows and columns of $X'$ that correspond to values greater than $m$,  the remaining entries exactly form the mechanism $X$ of Equation~(\ref{eq:mechanismsXY}). (Observe that such an extension of mechanism $X$ is indeed feasible, as any row which pertains to a value greater than $m$ can be identical to the row which pertains to the value $m$, and so the privacy constraints hold. Such a mechanism is also optimal, as the utility is a function of only the rows $\{0,1,\ldots,m\}$, due to the consumer's prior, so we cannot achieve a better utility than the utility gained by mechanism $X$). 
The second consumer believes that the query result cannot be larger than $m-1$  (in the Bayesian case, the consumer holds a uniform prior over $\{0,\ldots,m-1\}$). Refer to some optimal mechanism for this consumer as $Y'$. A similar argument shows, that ignoring rows and columns that pertains to values greater than $m$, the remaining entries match the mechanism $Y$ of Equation~(\ref{eq:mechanismsXY}). 
Assume for a contradiction that $X'$ and $Y'$ are both derivable from some mechanism $U'$. Therefore there exist remappings $T,S$ such that $X' = U'T$ and $Y'=U'S$. Let $U$ be the mechanism $U'$ reduced to only the inputs $\{0,1,\ldots,m\}$. Reducing $U'$ to $U$, we get that $X=UT$ and $Y=US$. According to the previous subsection these two mechanisms cannot be derived from a single oblivious mechanism, due to the same arguments in the proof of Claim \ref{clm:impossibility_sumquery}. Thus, a contradiction.

Now suppose for a contradiction that both the mechanisms are derived from a single {\em non}-oblivious mechanism $U^*$. This means that $U^*$'s input space corresponds to databases rather than to query results. Suppose there is a remapping $T$ such that $X^*=U^* T$. This means that the rows of $X^*$ correspond to databases as well. We assume that $X^*$ is oblivious (as universal optimality was shown not to exist even for count queries when consumers choose non-oblivious optimal mechanisms~\cite{GRS09}). Therefore, applying $U^*$ on two databases with the same query result and then applying T on $U^*$'s output,
yields identical rows in $X^*$ (which is described as a single row in the oblivious matrix $X$ above). Note that although $X^*$'s input and output spaces are discrete (and so we can refer to $X^*$ as a matrix), we assume nothing on $U^*$'s outputs and $T$'s inputs.
Reducing $U^*$ to an input space of only $m+1$ databases with different query results and applying the remapping $T$ on this reduced mechanism's output, yields mechanism $X$ completely. 
Similarly, applying some remapping $S$ on the same reduced mechanism yields mechanism $Y$. 
Now reduce $U^*$ to inputs which are the databases $(0,0,\ldots,0,q)$ where $q$ is any possible record value. Refer to this mechanism as $U$. According to the assumptions, we get that $X=UT,Y=US$. Also note that every two possible inputs of $U$ are neighboring databases, and so $U$ must satisfy privacy constraints as any oblivious mechanism. Therefore, we get a simple reduction to the case of an oblivious mechanism $U$, and the same impossibility result applies also to the case of non-oblivious universal mechanism\footnote{Actually, this also shows that enabling universal non-oblivious mechanisms cannot resolve such impossibility for every query whenever there are 3 (or more) values which are the exact query results of 3 different neighboring databases.}. Thus, we conclude Theorem \ref{thm:sum_queries}

\subsection{Impossibility of Universally Optimal Mechanisms for Histogram Queries}

The previous subsection shows that no universally optimal mechanisms exist for sum queries. In this and the following sections we consider other generalizations of count queries. One natural generalization is to histogram queries, and another is to bundles of simultaneous count queries. We begin with histogram queries. Note that a count query may be thought of as a histogram query where the database records are partitioned into two categories: those which satisfy a predicate, and those which do not. Consider now a histogram query which partitions the database records into three categories or more. 

\begin{mytheorem}
\label{thm:histograms}
No universally optimal mechanism exists for histogram queries, except for histograms for one predicate and its complement or trivial predicates. This holds both for the Bayesian and the risk-averse utility models.
\end{mytheorem}

\begin{proof}
Once again, consider first the case where there is only one record in the database, and the query is for a histogram which partitions the possible records into three categories. The only possible results for such a query are $\left(1,0,0\right)$, $\left(0,1,0\right)$ and $\left(0,0,1\right)$. Notice that all these histograms result from neighboring databases. Now consider information consumers whose loss function is either $\ell_{bin}$ or $\ell_1$ (in the case of a histogram over one record they both result with $0$ if the output matches the exact result and a constant otherwise). Refer to the first possible result as $0$, the second possible result as $1$, and the third possible result as $2$. Notice now that we have exactly the same constraints for valid mechanisms as we had for the sum query with just one record. Also, the utility expression for each of the consumers is the same. The problem of universally optimizing the utility for all $\ell_{bin}$ information consumers (or $\ell_1$ consumers) is now reduced to the same problem for sum queries. According to Subsection \ref{sec:sum_queries}, universally optimizing the utility for all such consumers is impossible, and so it is impossible to construct a mechanism for this specific case as well.

We now generalize this result for histograms over larger databases and partitions of any number of categories larger than $2$. First, consider the case of querying one record for a histogram of $c \geq 3$ categories. This can easily be reduced to problems we have already answered negatively. One way is to notice that as in the case where $c=3$ (in which we reduced this problem to the problem of sum queries where the records' values bound is $m=2$), larger values of $c$ can easily be reduced to sum queries with larger bounds on the records' values $m$. For every number of $c$ partitions, there are exactly $c$ possible results for the histogram over one record. They are all the results of neighboring databases. Refer to these results as $0,1,2,\ldots,c-1$. Again, this is exactly like constructing a universally optimal mechanism for sum queries over one record, in which the bound on its values is $m=c-1$. This is impossible as was shown in Subsection \ref{sec:m_gt_2}. Another way to be convinced is to refer to the partitions as $A_1, A_2, \ldots, A_c$. Now consider only consumers whose loss functions are depend on  the number of records in $A_1, A_2, (A_3 \cup A_4 \cup \ldots \cup A_c)$. These loss functions are monotone. This reduces the current problem to the problem of histograms over a partition of only $3$ categories, to which we already proved negative results.

We now generalize this result further to any size of the database. The same argument that was applied in Subsection \ref{sec:n_gt_1} for sum queries, applies here as well (even for the case of non-oblivious universal mechanisms). Consider only consumers with a prior such that all records except perhaps one fall into one specific category of the histogram. Querying for a histogram on such a database reduces to the result of the same histogram over one record only. Even if we consider only these consumers, we know that no one single mechanism can optimize their utilities over all possible mechanisms. Since there is no such mechanism that optimizes these consumers' utilities, there is obviously no mechanism that yields optimized results for all possible consumers. Therefore, even for larger databases, there is no universally optimal private mechanism for histogram queries.
\end{proof}

\subsection{Impossibility of Universally Optimal Mechanisms for Bundles of Count Queries}
\label{sec:bundles}

We now consider the generalization of single count queries to a bundle of count queries, where a bundle contains several simple (non trivial) count queries that need to be answered simultaneously. 
Note that a consumer's disutility for a bundle query need not be the sum of the losses for the separate basic queries -- it may be a more involved function of the bundle outputs. For instance, a consumer with the $\ell_{bin}$ loss function has no loss if all the results he uses are correct, and has one unit of loss if one or more of the results he uses are incorrect, no matter how many. Furthermore, information consumers may have auxiliary knowledge about the dependency between bundle outputs. 

\begin{mytheorem}
\label{thm:bundles}
No universally optimal mechanism exists for bundles of more than one simultaneous non-trivial count queries. This holds both for the Bayesian and the risk-averse utility models.
\end{mytheorem}

\begin{proof}
Such a generalization of count queries proves to be no different than the other intuitive generalizations we have already discussed. Note that two simultaneous non-trivial different predicates actually partition the records domain into $4$ categories: those which satisfy both predicates, those which satisfy none of them, and those which satisfy just the first or just the second. 
If the predicates are somehow related, then the predicates might partition the domain into only $3$ categories. This may happen in various cases, namely if one of the predicates is a subset of the other, if no record can possibly satisfy both of the predicates, or if any possible record must satisfy at least one of the predicates.
Either way, there are always three different outputs for such bundles which result from three neighboring databases. (This is of course true also if the bundle consists of more than two simultaneous count queries). Once more, consider two different information consumers. The first has the $\ell_{bin}$ loss function and a uniform prior over these three outputs (Resp.~in the risk-averse model, her support is the set of these three outputs). The second consumer also uses the $\ell_{bin}$ loss function and has a uniform prior over two of these outputs. (Resp.~in the risk-averse model, her support is a set of two of these three outputs). Name these different outputs $0$, $1$ and $2$. As in the previous subsection, the problem of universally optimizing the utility for all $\ell_{bin}$ information consumers is now reduced to the same problem presented in Subsection \ref{sec:sum_queries} . (The constraints for valid mechanisms are the same, and the utility expression for each of the consumers is the same). The only optimal mechanisms for the chosen information consumers are the same as those in Observations \ref{obsrv:Bayesian} and \ref{obsrv:risk-averse}. According to Claim \ref{clm:impossibility_sumquery}, such mechanisms are not derivable from any single private mechanism, and so universally optimizing the utility for all such consumers is impossible in the queries bundles as well.

\remove{
Consider first the case where the predicates partition the domain into only $3$ categories. In that case determining the results of the two count queries also determines the histogram values over the specified partitions (the size of the database is known), and vice versa. Therefore, any prior of an information consumer over the results of these two queries can be immediately translated into prior over the matching histogram, and vice versa. We focus on information consumers with the $\ell_{bin}$ loss function, which means the consumer pays no penalty if she uses the correct results, and pays a constant penalty if using wrong results. This loss function is defined the same way in every scenario. So, in that case, all the parameters of the consumer (prior, loss function and target function) have a simple translation from the histogram scenario to the bundle of queries scenario, and vice versa. In fact, these are the query results and the histogram are just two different ways of expressing the exact same information, and so, every mechanism in one scenario can be translated into a mechanism in the other scenario. Since we have already seen in the previous subsection that there are two information consumers with the $\ell_{bin}$ loss function who cannot be optimally satisfied by a single private mechanism which answers the given queries, then there is no universally optimal mechanism for answering this bundle of count queries as well.

Now consider the case where the predicates partition the domain into $4$ categories. The proof will be quite the same. Notice, though, that knowing the histogram values implies knowing the queries results, but not necessarily the other way around. Therefore, we will explicitly stipulate two information consumers who cannot be optimally satisfied by a single mechanism. Consider the 3 cases where the count query results are both $0$, the first is $0$ and the second is $1$, the first is $1$ and the second is $0$. All these 3 possible results ......

If one of the predicates is a subset of the other, then the predicates partition the domain into only $3$ categories. In each case, the reply for such two count queries is the same as the reply for a histogram query with $3$ or $4$ categories. We have already shown that no universally optimal mechanism exists for histogram queries (which are not identical to a single count query). Therefore, no universally optimal mechanism exists for a bundle of two count queries, as well. The same argument shows that no universally optimal mechanism exists for a bundle of any number (larger than 1) of different non-trivial count queries.
}
\end{proof}

%The proofs for these theorems appear in Appendix \ref{app:histograms}.
%The proof for these theorems is quite similar to the proof given for the case of sum queries, and its main components were supplied in the proof for sum queries. The full discussion is postponed to Appendix \ref{app:histograms}.

\section{A Characterization of Universal Optimality in the Bayesian Setting} 
\label{sec:characterization}

We now discuss a more general setting, where a query (not necessarily related to sum or count) is answered by a differentially private mechanism in the Bayesian utility model. We follow other works on this subject and only consider oblivious private mechanisms.
Note that although our results do not exclude the possibility of non-oblivious differentially private mechanisms, our techniques yield that no such non-oblivious universally optimal mechanisms exist for many natural functions. Specifically, enabling universal non-oblivious mechanisms cannot resolve such impossibilities for a query whenever there are 3 (or more) values which are the exact query results of 3 different neighboring databases. This is due to the same argument that was used in Subsection \ref{sec:n_gt_1}.

Let the database records be taken from a discrete domain $\D$ and let the query be $f:\D^n\rightarrow\R_f$ (wlog, we will assume that $f$ is a surjective function, in which case $\R_f = \{f(D): D\in\D^n\}$ is also a discrete set). Define the following  graph where edges correspond to answers $f$ may give on neighboring databases (and hence to restrictions on output distributions implied by differential privacy):
\begin{mydefinition}[Privacy Constraint Graph]
\label{def:graph}
Fix a query $f:\D^n\rightarrow\R_f$. The \emph{Privacy Constraint Graph} for $f$ is the undirected graph $G_f=(V,E)$ where $V=\R_f$ is the set of all possible query results and \mbox{$E = \{(f(D_1),f(D_2)) : D_1,D_2\in\D^n~\mbox{are neighboring}\}$}. The \emph{degree} of the constraint graph, $\Delta(G_f)$, is the maximum over its vertices' degrees. For $i_2,i_2\in\R_f$,  $G_f$ induces a distance metric $d_{G_f}(i_1,i_2)$ that equals the length of the shortest path in $G_f$ from $i_1$ to $i_2$.
\end{mydefinition}

Observe that the constraint graph is connected for any query $f$: If $i_1=f(D_1)$ and $i_2=f(D_2)$ then there is a sequence of neighboring databases starting with $D_1$ and ending in $D_2$, and hence a path from $i_1$ to $i_2$ in $G_f$. 

Recall that the results of \cite{GRS09,GS10} are restricted to  loss functions $\ell(i,r)$ that are monotonically non decreasing in the metric $|i-r|$. In our more general setting, we avoid interpreting outcome of $f$ as points of a specific metric space, and hence we only consider the $\ell_{bin}$ loss function, which would remain monotone under any imposed metric.

\paragraph{Outline of this Section.} We are now ready to describe the results of this section. Let $f$ be a query, and $G_f$ its constraint graph. We first show that if $G_f$ is a single cycle, then no universally optimal mechanism exists for $f$. This impossibility result is then extended to the case where $G_f$ contains a cycle.
\begin{mytheorem}
\label{thm:acyclic}
Fix a query $f:\D^n\rightarrow\R_f$, and let $G_f$ be its constraint graph. Consider Bayesian information consumers with loss function $\ell_{bin}$.  If $G_f$ contains a cycle then no universally optimal mechanism exists for these consumers.
\end{mytheorem}
Constraint graphs of sum queries (for $m \geq 2$), histograms and bundles of queries all have cycles of length $3$, so, in the Bayesian utility model, Theorem~\ref{thm:acyclic} generalizes all our previous results.

Next, we consider the case where $G_f$ is a tree and show that if $G_f$ contains a vertex of degree $3$ or higher, then no $\alpha$-differentially private universally optimal mechanism exists for $f$  for $\alpha> 1/(\Delta(G_f)-1)$. (Recall that the closer $\alpha$ is to one, the better privacy we get.)
\begin{mytheorem}
\label{thm:delta3}
Fix a query $f:\D^n\rightarrow\R_f$, and let $G_f$ be its constraint graph. Consider Bayesian information consumers with loss function $\ell_{bin}$.  If the privacy parameter $\alpha > 1/ (\Delta(G_f)-1)$ then no universally optimal mechanism exists for these consumers.
\end{mytheorem}
We can conclude  from theorems~\ref{thm:acyclic} and~\ref{thm:delta3} that for $\alpha > 0.5$, the only functions $f$ for which universally optimal mechanisms exist are those where $G_f$ is a simple chain, as is the case for the count query.

The proof structure is similar to the one presented in the previous section for sum queries. We begin with the case where $G_f$ is a simple cycle. We consider two consumers with different priors and loss function $\ell_{bin}$, and show that the optimal mechanisms for these consumers must have specific structures (in the sense that some privacy constraints are satisfied tightly). Once again, we show that for two mechanisms with such structures, there is no mechanism which is at least as general as these two (i.e., there is no single mechanism which derives both of them).

Next, we extend the proof to the case where $G_f$ contains a cycle. We focus on a cycle in $G_f$ of smallest size $m$, and consider two information consumers. The consumers are similar to those for the case where $G_f$ is a cycle, and so are the optimal mechanisms for them, except that we need to prove that these optimal mechanisms can be extended in a differentially private manner to the entire range of $f$. For that we introduce a labeling of $G_f$ in which the labels of adjacent vertices differ by at most one modulo $m$. 

Last, we discuss the case where $G_f$ is a tree containing a vertex of degree at least $3$. Focusing on that vertex and three of its adjacent vertices, we present three consumers with different priors. Again, we focus on the corresponding entries in the matrices of their optimal mechanisms, and find which constraints must be tight. Assuming all three mechanisms are derived from a single mechanism $U$, we present three different partitions of $U$'s range according to which constraints are tight for every measurable subset of $U$'s range. Combining the attributes from these partitions, we get one elaborated partition of $U$'s range. We can then assume $U$'s range is finite and reveal the structure of its matrix columns. Such a structure of $U$'s columns (for the consumers we chose) is feasible iff we compromise for a privacy parameter $\alpha \leq 0.5$. Finally, we generalize this claim to any degree of one vertex.

\subsection{The Basic Case: $G_f$ is a Cycle}
\label{subsection:basic_case_cycle}

We begin with the simple case where $G_f$ is a single cycle of $m>2$ vertices%
\footnote{An example query that yields such a graph is $f: \{0,1\}^n \rightarrow [m]$ defined as $f(d_1,\ldots,d_n) = \sum_{i=1} d_i \mod m$. If $n\geq m >2$ then $G_f$ is a cycle of size $m$.}. 

\begin{myclaim}\label{clm:Gf_cycle}
If the constraint graph $G_f$ of $f:\D^n\rightarrow\R_f$ is a single cycle, then no universally optimal mechanism for Bayesian information consumers exists for $f$.
\end{myclaim}
\begin{proof}
Assume $G_f$ is the cycle $C_m = (v_0,v_1,\ldots,v_{m-1},v_0)$. We already proved impossibility of universal optimality for the case $m=3$ in Claim~\ref{clm:impossibility_sumquery}. We now deal with the case $m > 3$. As in the proof of Claim~\ref{clm:impossibility_sumquery}, we will present two information consumers, and their corresponding optimal mechanisms, and prove that these cannot be derived from a single mechanism.

We first consider an information consumer with loss function $\ell _{bin}$ and prior $p_{v_0}=p_{v_1}=\cdots = p_{v_{m-1}}=1/m$, and construct the unique optimal mechanism $X$  for this consumer. ($X$ is represented by an $m\times m$ matrix since with the $\ell_{bin}$ loss function the support of the optimal mechanism's range must match the support of the consumer's prior.) An optimal mechanism minimizes
$$
\sum_{v_i \in C_m} p_{v_i} \sum_{r \in C_m} x_{v_i,r} \cdot \ell_{bin} (v_i,r)
 = \sum_{v_i \in C_m} p_{v_i} \cdot (1-x_{v_i,v_i})         
 = 1 - \frac{1}{m} \sum_{v_i \in C_m} x_{v_i,v_i},
$$
and hence, the consumer's goal is to maximize $\sum_{v_i \in C_m} x_{v_i,v_i}$ subject to maintaining $\alpha$-differential privacy. Maintaining $\alpha$-differential privacy implies
\begin{equation}\label{eq:dp_cycle}
\alpha^{d_{G_f}(v_i,v_j)} x_{v_i,v_i}  \leq x_{v_j,v_i}  \quad \forall v_i,v_j \in C_m ,
\end{equation}
and hence, by summing up the inequalities for all $v_i,v_j$, we get
$$\sum_{v_i \in C_m} \sum_{\substack{v_j \in C_m \\ v_j \neq v_i}}  \alpha^{d_{G_f}(v_i,v_j)}  x_{v_i,v_i} \leq \sum_{v_i \in C_m} \sum_{\substack{v_j \in C_m \\ v_j \neq v_i}} x_{v_j,v_i}   
= \sum_{v_i \in C_m} 1 - x_{v_i,v_i}    
= m - \sum_{v_i \in C_m} x_{v_i,v_i},
$$
and we conclude that  
$$\sum_{v_i \in C_m} x_{v_i,v_i} \leq \frac{m}{1+\sum_{\substack{v_j \in C_m \\ v_i \neq v_j}} \alpha^{d_{G_f}(v_i,v_j)}}. $$
This inequality is tight iff Equation~(\ref{eq:dp_cycle}) is tight (i.e., $x_{v_j,v_i} = \alpha^{d_{G_f}(v_i,v_j)}x_{v_i,v_i}$) for every $v_i \neq v_j \in C_m$. In such a case, we can find the mechanism's entries by solving a system of $m$ linear equations  (the sum of each row in the mechanism must be $1$), in a similar argument to the one presented in the proof of Observation~\ref{obsrv:Bayesian}. Since these are $m$ independent linear equations in $m$ variables, our optimal solution for $x_{v_1,v_1},\ldots ,x_{v_m,v_m}$ is unique. 

Utilizing the symmetry of the equations, we get that every row of $X$ is a cyclic shift of:
\begin{align}
\label{eq:X_rows}
& \delta \cdot (1, \alpha^1, \alpha^2, \ldots, \alpha^{(m-1)/2},\alpha^{(m-1)/2},\alpha^{(m-1)/2-1},\ldots, \alpha^2,\alpha^1) \quad & \text{if}~m~\text{is odd,} \\
& \delta \cdot (1, \alpha^1, \alpha^2, \ldots, \alpha^{m/2-1},\alpha^{m/2},\alpha^{m/2-1},\ldots, \alpha^2,\alpha^1) & \text{if}~m~\text{is even}.  \nonumber
\end{align}
where $\delta$ is chosen such that $X$ is row-stochastic. The mechanism $X$ satisfies $\alpha$-differential privacy, it is optimal for our information consumer, and it is unique.

Our second information consumer uses $\ell_{bin}$ as her loss function, and prior \mbox{$p_{v_0} = p_{v_1} = p_{v_2} = 1/3$} and $p_{v_3}=\cdots=p_{v_{m-1}}=0$. Note that since $m>3$ the vertices $v_0,v_2$ are not adjacent in $G_f$ (so \mbox{$d_{G_f}(v_0,v_2)=2$}). In constructing an optimal mechanism $Y$ for the information consumer we will only consider the rows and columns pertaining to vertices $v_0,v_1,v_2$, noting that the columns for all other vertices contain only zeros, and there is some freedom with respect to the rows for the other vertices. Applying similar arguments as for mechanism $X$, we get that the columns of $Y$ are of the forms $(1,\alpha^1,\alpha^2)^T$, $(\alpha^1,1 ,\alpha^1)^T$, $(\alpha^2,\alpha^1,1)^T$ (each of the columns may be multiplied by a different coefficient). By forcing row stochasticity,  we can solve the following equations to get the coefficients:
\begin{equation*}
\begin{bmatrix}
	1        & \alpha^1 & \alpha^2 \\
	\alpha^1 & 1        & \alpha^1 \\
	\alpha^2 & \alpha^1 & 1       
\end{bmatrix}
\times
\begin{pmatrix}	c_1 \\	c_2 \\	c_3 
\end{pmatrix}
=  \begin{pmatrix}	1 \\ 1 \\ 1 \end{pmatrix}
\end{equation*}
and we get a unique structure on the entries of these rows and columns of $Y$. This mechanism is of no surprise, as these entries are merely the finite-range version of the geometric mechanism (as shown in \cite{GRS09}).

Summarizing our findings, we get that
\begin{equation}
\label{eq:optimal_mecs}
X = \delta \cdot 
	\begin{bmatrix}
		1 & \alpha & \alpha^2 & \cdots & \alpha^2 & \alpha \\
		\alpha & 1 &\alpha  & \cdots & \alpha^3 & \alpha^2 \\
		\alpha^2 & \alpha & 1 & \cdots & \alpha^4 & \alpha^3 \\
		\vdots & \vdots & & \ddots & & \vdots \\
		\alpha & \alpha^2 &\alpha^3 &\cdots &\alpha &1
	\end{bmatrix};\;
Y =\begin{bmatrix}
	c_1        & c_2 \cdot\alpha^1 & c_3 \cdot\alpha^2 & 0 & \cdots & 0\\
	c_1 \cdot\alpha^1 & c_2        & c_3 \cdot\alpha^1 & 0 & \cdots & 0\\
	c_1 \cdot\alpha^2 & c_2 \cdot\alpha^1 & c_3       & 0 & \cdots & 0 \\
	\vdots & \vdots & \vdots & \vdots& \ddots  & \vdots \\
\vdots & \vdots & \vdots & 0 & \cdots & 0 
\end{bmatrix}.
\end{equation}

We now show that instances of such mechanisms $X$ and $Y$ are not derivable from a single mechanism. Since the conditions stated for these mechanisms are necessary for them to be optimal for the two consumers we chose, this will prove that there is no universally optimal mechanism in such a scenario. 

Suppose, towards a contradiction, that there exists a mechanism $U$ which derives both $X$ and some instance of $Y$. According to the characterization of generally maximal differentially private mechanisms (Theorem \ref{thm:maximal_characterization}), $X$ is maximally general. Therefore, we get that $U$ is derivable from $X$ and so $Y$ is derivable from $X$ as well. Therefore, there exists a remapping matrix $T$ such that $Y = XT$. 
Remember that $Y$'s columns are linear combinations of $X$'s columns with non-negative coefficients, as described in the proof of Claim \ref{clm:impossibility_sumquery}. Any tight constraint met in one of $Y$'s columns must match the same tight constraints in all of $X$'s columns which appear in the linear combination of that column. Once again, any specific column of $X$ must appear in at least one linear combination of one of $Y$'s columns with a positive coefficient (as any possible output of $X$ must be remapped to the values $\{v_0,v_1,v_2\}$ by $T$).
Notice that one of $X$'s columns is 
\begin{align*}
& \delta \cdot (\alpha^{(m-1)/2}, \alpha^{(m-1)/2},\alpha^{(m-1)/2-1},\ldots , 1 , \ldots , \alpha^{(m-1)/2-1})^T \quad & \mbox{if $m$ is odd,} \\
& \delta \cdot (\alpha^{m/2-1}, \alpha^{m/2},\alpha^{m/2-1},\ldots , 1 , \ldots , \alpha^{m/2-2})^T & \mbox{if $m$ is even},
\end{align*}
Mapping this column into any of $Y$'s first three columns (with any positive probability) cannot yield the tight constraints which appear in the first three entries of the chosen column in $Y$. Therefore, no such remapping $T$ is feasible and we get a contradiction.

\remove{
Suppose, towards a contradiction, that there exists a mechanism $U$ which induces $X$ and some instance of $Y$. Then there exist remappings $T,S$, such that $X= T \circ U$ and $Y= S \circ U$. For simplicity of the presentation, we assume $U$ is discrete (the proof can be extended for mechanisms with continuous range similarly to the extension of Claim~\ref{clm:impossibility_sumquery} in Appendix~\ref{app:continuous_range}). Every column in $X$ has $m-1$ tight constraints  in case $m$ is odd, and $m$ tight constraints in case $m$ is even (see Equation~(\ref{eq:X_rows})). Define $B$ to be all the zero columns of $U$. 
Define 
\begin{equation*}
	A_i = \{ u \in \mbox{Range}(U) : \mbox{$u$ is remapped with a positive probability to $v_i$ by mechanism $T$} \} \setminus B.
\end{equation*}
Similar arguments to those presented in the proof of Claim~\ref{clm:impossibility_sumquery} show the sets $A_1,A_2,\ldots A_m,B$ partition the columns of $U$ ($B$ may be empty, but none of the sets $A_i$ is empty), and that all columns of a set $A_i$ follow the same structure of tight constraints as the relevant column in the derived mechanism $X$. I.e., all the columns of $A_1$ are of the form 
$ \delta_1 \cdot (1,\alpha^1,\alpha^2,\alpha^3,\ldots , \alpha^2,\alpha^1)^T$,
all the columns of $A_2$ are of the form 
$ \delta_2 \cdot (\alpha^1 ,1,\alpha^1,\alpha^2,\ldots , \alpha^3, \alpha^2)^T$, and so on, where $\delta_1,\delta_2,\ldots > 0$. 
In other words, (up to the multiplicative coefficients) columns in $A_i$ are like columns in $A_{i-1}$ cyclically shifted one position. In particular, there must exist a column $z$ in $U$ that is of the form 
\begin{align*}
& \delta \cdot (\alpha^{(m-1)/2}, \alpha^{(m-1)/2},\alpha^{(m-1)/2-1},\ldots , 1 , \ldots , \alpha^{(m-1)/2-1})^T \quad & \mbox{if $m$ is odd,} \\
& \delta \cdot (\alpha^{m/2-1}, \alpha^{m/2},\alpha^{m/2-1},\ldots , 1 , \ldots , \alpha^{m/2-2})^T & \mbox{if $m$ is even},
\end{align*}
where $\delta > 0$.

 Consider now a similar partitioning according to the constraints of mechanism $Y$, and remember that the first three entries of this column are entries in the rows of $v_0,v_1,v_2$ (i.e., we now have constraints only on the ratios between the first three entries in each column). We get that every non-zero column of $U$ is of one of the forms: $\gamma_1 \cdot (1, \alpha^1,\alpha^2, \ldots)^T, \gamma_2 \cdot (\alpha^1, 1, \alpha^1, \ldots)^T, \gamma_3 \cdot (\alpha^2, \alpha^1, 1, \ldots)^T$. This leads to a contradiction as the column $z$ does not match any of these subsets of $U$'s columns.
}

\end{proof}

\subsection{Impossibility of Universal Optimality When $G_f$ Contains a Cycle}

We now give a proof for Theorem \ref{thm:acyclic} which deals with the case where $G_f$ \emph{contains} a cycle.

\begin{proof}
Let $C_m=(v_0,v_1,\ldots ,v_{m-1},v_0)$ be a cycle of smallest size in $G_f$. Based on $C_m$, we will consider two consumers whose optimal mechanisms contain as sub-matrices the matrices  $X,Y$ from the proof of  Claim~\ref{clm:Gf_cycle}, and hence they cannot be derived from a single mechanism.

\paragraph{The First Consumer: uniform prior over $C_m$}

Consider an information consumer with loss function $\ell _{bin}$ and prior $p_{v_0}=p_{v_1}=\cdots = p_{v_{m-1}}=1/m$ and $p_u=0$ for every $u \notin C_m$. We will construct an optimal mechanism $X'$ for this consumer, and will prove that (in some sense) it is unique.
We begin with a labeling algorithm of the vertices in $G$:
\begin{enumerate}
\item \label{step:labelC_m} Given $C_m=(v_0,v_1,\ldots ,v_{m-1},v_0)$, set $l(v_i) = i$ for $i\in \{0,\ldots,m-1\}$.
\item \label{step:label_s} For $s$ from $1$ to $m-1$: 
\begin{enumerate}
\item Let $V_s$ be the set of unlabeled vertices that are adjacent to vertices labeled $s-1$.
\item Let $l(u)=s$ for all $u\in V_s$.
\end{enumerate}
\item \label{step:label_all_rest} Let $l(u) = m-1$ for all remaining vertices $u$.
\end{enumerate}

\begin{myclaim}
After applying the above algorithm, the labels for every two adjacent vertices differ by at most 1 (modulo $m$).
\end{myclaim}
\begin{proof}
We show that at any stage of the labeling, any two adjacent vertices satisfy the requirement that their labels differ by at most 1 (modulo $m$).

Note first that this holds for all labeled vertices after Step~\ref{step:labelC_m}.
Consider a vertex $u\in V_s$ (i.e., $l(u)=s$ is set in iteration $s$), and an adjacent vertex $u'$ that is labeled $l(u')=s'$ prior to or on iteration $s$.
Clearly, if $u'\in V_s \cup V_{s-1}$ then $s'\in\{s-1,s\}$ and the statement holds for $(u,u')$. 
Otherwise, we consider two sub-cases. In the first, $l(u') = s' < s-1$, and we are led to a contradiction since $u$ remains unlabeled after iteration $s'+1$ whereas by definition $u \in V_{s'+1}$. In the second sub-case $l(u') = s' > s+1$ (if $s'=s+1$ the claim holds) and hence it must have been that $u'$ was labeled in Step~\ref{step:labelC_m}, i.e., $u' = v_{s'}$ for $s' \in \{s+1,\ldots, m-1\}$. Following the path of labels which led to the label of $u$ we can get to the vertex $v_0$ via a path of length $s$. Noting that this path is disjoint from the length $m-s'$ path $v_{s'}\leadsto v_0 = v_{s'},v_{s'+1},\ldots,v_{m-1},v_0$, we get that $G$ contains the cycle $v_{s'} \leadsto v_0 \leadsto u \leadsto u'$ that is of length $m-s' + s + 1 < m$, in contradiction to $C_m$ being the smallest cycle in $G$. 
To conclude the proof, note that every vertex $u \in G$ adjacent to some $u' \in G$ such that $l(u') \in \{0,1,\ldots , m-2\}$ has been labeled in iteration $l(u')+1$ or earlier. Therefore in Step~\ref{step:label_all_rest}, the vertices which are not labeled yet are adjacent only to unlabeled vertices and to vertices with label $m-1$. Labeling the remaining vertices with $m-1$ satisfies the requirement. 
\end{proof}

We now use the graph labels to construct an optimal mechanism $X'$, represented by a matrix of dimensions $|\R_f| \times |\R_f|$. The entries of rows $u \notin C_m$ have no effect on the expected loss of this consumer, as $p_u = 0$. There are, however, restrictions on these rows, as the mechanism $X'$ must be differentially private. We construct $X'$ as follows:
\begin{enumerate}
\item For all $u \notin C_m$, set column $u$ of $X'$ to be a column of zeros.
\item For all $u \in C_m$, set row $u$ of $X'$ as in the optimal mechanism $X$ described in the proof of Claim~\ref{clm:Gf_cycle} (i.e., Equation~(\ref{eq:optimal_mecs})).
\item For all $u \notin C_m$, set row $u$ of  $X'$ to be identical to the row corresponding with the vertex identically labeled in $C_m$.
\end{enumerate}
Clearly, the resulting mechanism is row-stochastic. The privacy constraints also hold: suppose \mbox{$u,u' \in \R_f$} are query results of neighboring databases. Therefore, they are adjacent in the constraint graph, and their labels differ by at most $1$ (modulo $m$). And so, their matching rows in mechanism $X'$ are either identical or they are the same as rows of two adjacent vertices $v_i,v_j \in C_m$ in mechanism $X$. Since the construction of rows in $X$ hold to the privacy constraints, so do the rows of $X'$. In other words, we just showed that mechanism $X$ can be extended to any query $f$ whose constraint graph $G_f$ contains $C_m$ but no smaller cycles. 

Notice that only rows of $C_m$ affect the expected loss in $X'$, which is hence identical to that of $X$. Since any mechanism in this scenario has to satisfy all the restrictions for just the vertices of the cycle $C_m$, and more, the expected loss for any optimal mechanism in the current scenario is lower bounded with that of $X$. Hence, we can conclude that $X'$ is optimal for the information consumer, and furthermore, $X'$ restricted to the rows corresponding to $C_m$ is unique.

\paragraph{The Second Consumer: uniform prior over $v_0,v_1,v_2$}
Consider an information consumer with loss function $\ell _{bin}$ and prior $p_{v_0}=p_{v_1}=p_{v_2}=1/3$ and $p_u=0$ for every other $u \in \R_f$. We argue that every optimal mechanism $Y'$ for this consumer has the same structure on rows $v_0,v_1,v_2$ as mechanism $Y$ in Equation~(\ref{eq:optimal_mecs}). As the impossibility of universal optimality for the case of $m=3$ was already covered, and we assumed $m>3$, $v_0$ and $v_2$ are not adjacent in $G_f$. This enables us to label the vertices in such a way: $l(v_0)=0$, $l(v_2)=2$ and $l(u)=1$ for any other vertex in $G_f$. Again, it is clear that every two adjacent vertices have labels which differ by $1$ at most. Similar arguments as the ones presented for the first consumer, show that the first three rows of every optimal mechanism for this consumer (i.e.~the rows for $v_0,v_1,v_2$) have the same structure as the first three rows of mechanism $Y$ in Equation~(\ref{eq:optimal_mecs}).

Assume towards a contradiction that both $X'$ and $Y'$ are derivable from a single mechanism $U'$. Therefore there exist remappings $T,S$ such that $X'=U'T$ and $Y'=U'S$. Let $U$ be the mechanism $U'$ reduced to only the inputs of the cycle $C_m = \{v_0,v_1,\ldots ,v_{m-1}\}$. Reducing $U'$ to $U$, we get that $X=UT$ and $Y=US$. According to the previous subsection these two mechanisms cannot be derived from a single
oblivious mechanism, due to the same arguments in the proof of Claim \ref{clm:impossibility_sumquery}. Thus, we get a contradiction. 

\remove{
Ignoring the rows of query results other than the $C_m$ vertices, we can partition once more the columns of every universal mechanism $U$ according to the columns of $X'$ and according to the columns of $Y'$. The same argument, as presented in the basic case of a single cycle graph, shows that such two partitions contradict, and there can be no universal mechanism which induces optimal instances of $X'$ and $Y'$.
}
\end{proof}

\subsection{Impossibility of Universal Optimality When $\Delta(G_f) \geq 3$}

We now focus on acyclic constraint graphs and prove Theorem \ref{thm:delta3} and its conclusion that for $\alpha > 0.5$ no universally optimal mechanisms exists unless the constraint graph is a simple chain.

\begin{proof}
For simplicity of this proof, we first focus only on $3$ neighbors of a specific vertex, and prove that no universally optimal mechanism exists for $\alpha>1/(3-1)=0.5$. Later, we generalize this result for a vertex of any degree by taking into account all of the vertex's neighbors. The generalization is done using the same methods we use to prove the simpler case.

Let $v_0$ be a vertex in $G_f$ with a degree greater than $2$. Let $v_1,v_2,v_3$ be $3$ of its neighbors. We choose some consumers with loss function $\ell_{bin}$ and zero a priori probability for all values other than $v_0,v_1,v_2,v_3$. We define some necessary conditions on the optimal mechanisms of these consumers and show it is impossible to simultaneously derive optimal mechanisms for these consumers from a single mechanism $U$ (when $\alpha > 0.5$). 

Note that by the tree structure of $G_f$, every mechanism that satisfies the requirements of differential privacy on query results $v_0,v_1,v_2,v_3$ can be easily extended to a differentially private mechanism on all results of $\R_f$%
\footnote{One possible extension is as follows: Suppose $X$ is a mechanism from $\{v_0,v_1,v_2,v_3\}$ to $\{v_0,v_1,v_2,v_3\}$. Label each of the vertices $v_0,v_1,v_2,v_3$ by $l(v_i)=i$, then label every other vertex in the graph with the same label as its nearest labeled vertex. Construct a mechanism $X'$ from $X$ like this: Set $x'_{v_i,v_j} = x_{v_i,v_j}$ for every $i,j \in \{0,1,2,3\}$. Set $x'_{v_i,u}=0$ for every $u \notin \{v_0,v_1,v_2,v_3\}$. For every $u \notin \{v_0,v_1,v_2,v_3\}$, set the row of $u$ to be the same as the row of $v_{l(u)}$. }. 
Furthermore, since our consumers have zero a priori probability for all other values, the entries in rows corresponding to values other than $v_0,v_1,v_2,v_3$ do not affect the consumers' expected loss.
Hence, it suffices to show the impossibility result for the case where $G_f$ is restricted to $v_0,v_1,v_2,v_3$.%
\footnote{An example query that yields such a graph is $f: \{1,2,3\}^n \rightarrow \{0,1,2,3\}$ defined as $f(D) = i$ if all records in $D$ equal $i$, $0$ otherwise.}

Consider first an information consumer with prior $p_{v_0}=p_{v_1}=p_{v_2}=1/3, p_{v_3}=0$. note that $v_1,v_0,v_2$ is a simple path of length $3$ in $G_f$ for which the optimal mechanism was described in Section~\ref{subsection:basic_case_cycle}. Any optimal mechanism for this consumer is of the form 
\begin{equation*}
Y=
\begin{bmatrix}
	c_0                & c_1 \cdot \alpha    &  c_2 \cdot \alpha    & 0 \\
	c_0 \cdot \alpha   & c_1                 &  c_2 \cdot \alpha^2  & 0 \\
	c_0 \cdot \alpha   & c_1 \cdot \alpha^2  &  c_2                & 0 \\
	q_0  & q_1 &  q_2  & 0 
\end{bmatrix}
\end{equation*}
where $q_0+q_1+q_2=1$ and they are subject to some privacy constraints. The restrictions on the optimal mechanism for this consumer are the same as those on mechanism $Y$ in Equation~(\ref{eq:optimal_mecs}), only now $v_0$ is the vertex in the middle, so the first two rows were swapped, as were the first two columns.

Suppose that such a mechanism was derived from a universal mechanism $U$ by some remapping $T$. Suppose for now that $U$'s range is discrete and so it can be expressed in matrix form. (We abuse a little the notion of a matrix and allow $U$ to have infinitely many columns, and $T$ to have infinitely many rows, if needed). As noted before, this means that $Y$'s columns are linear combinations with positive coefficients of columns in $U$. Also, remember that since the coefficients are non-negative, linearly combining columns which do not hold tight privacy constraints, cannot yield a column with tight constraints. Since $T$ is row-stochastic, every row of $T$ has at least one positive entry. This means that every column in $U$ is remapped (with some positive probability) to a column in $Y$. Assume $U$ does not have zero columns (otherwise we could just ignore them as they pertain to results which are not in $U$'s range). From the reasons above and the the structure of constraints in $Y$ which are tight, we conclude that all of $U$'s columns can be partitioned into columns of the forms: $\delta_1 \cdot(1, \alpha, \alpha, \ast)^T$, $\delta_2 \cdot(\alpha, 1, \alpha^2, \ast)^T$, $\delta_3 \cdot(\alpha, \alpha^2, 1, \ast)^T$. The first set of columns is summed by $T$ into the first column of $Y$, The second set is summed by $T$ into the second column of $Y$, and the third set is summed to the third column of $Y$. The $\ast$ can take infinitely many values as it does not necessarily match to a tight constraint in $Y$. 

Considering now an information consumer with a prior $p_{v_0}=p_{v_1}=p_{v_3}=1/3,p_{v_2}=0$, and applying the same arguments, we have that the non-zero columns of the universal mechanism $U$ are partitioned into columns of the forms $\delta_1 \cdot(1, \alpha, \ast, \alpha)^T, \delta_2 \cdot(\alpha, 1, \ast, \alpha^2)^T, \delta_3 \cdot(\alpha, \alpha^2, \ast, 1)^T$. Similarly, considering a consumer with a prior $p_{v_0}=p_{v_2}=p_{v_3}=1/3,p_{v_1}=0$, we have that the non-zero columns of the universal mechanism $U$ are partitioned into columns of the forms $\delta_1 \cdot(1, \ast, \alpha, \alpha)^T, \delta_2 \cdot(\alpha, \ast, 1, \alpha^2)^T, \delta_3 \cdot(\alpha, \ast, \alpha^2, 1)^T$.

Notice that every non-zero column in $U$ must match one category in each of the partitions described above. Combining these conditions together, we have that the non-zero columns of $U$ are partitioned into columns of the forms $\gamma_1 \cdot (1,\alpha, \alpha, \alpha)^T, \gamma_2 \cdot (\alpha, 1, \alpha^2, \alpha^2)^T, \gamma_3 \cdot (\alpha, \alpha^2, 1, \alpha^2)^T, \gamma_4 \cdot (\alpha, \alpha^2, \alpha^2, 1)^T$.
As the columns in every category are proportional to one another, we assume that the mechanism $U$ has exactly one column in each of these categories. We can assume that, since if $U'$ is a mechanism with two non-zero columns which are proportional to one another, we can produce a mechanism $U$ by replacing these columns with a single column containing their sum. Then $U$ is derivable from $U'$, and vice versa. Therefore these mechanisms are equivalent.

Note that we assumed $U$'s range is discrete only for convenience. $U$'s range can be continuous as well, as explained by Kifer and Lin \cite{KL10}. Define $T$'s inverse to be for every vertex $v$, $T^- (v) = \{o' \in Rng(U) : Pr[T(o')=v] > 0 \}$. The same arguments from before hold, and we get that for every measurable $O' \subseteq T^-(v)$, and any adjacent vertices $v_i,v_j$, $\frac{Pr[T \circ U (v_i)=v]}{Pr[T \circ U (v_j)=v]} = \frac{Pr[U(v_i) \in O']}{Pr[U(v_j) \in O']}$, unless one of the probabilities is zero in which case all the probabilities are zero due to differential privacy constraints. This is because we assume $\frac{Pr[T \circ U (v_i)=v]}{Pr[T \circ U (v_j)=v]}$ is tight by differential privacy constraints,  and it can be expressed as a positive combination over measurable sets in $T^-(v)$ which, therefore, must be tight as well. And so, the same structure of tight constraints as they appear in the derived mechanism, must appear also for every measurable subset of $T^-(v)$ for every $v$ in the derived mechanism's output.

We conclude that a universal mechanism must be of the form:
\begin{equation*}
U=
\begin{bmatrix}
	c_0               & c_1 \cdot \alpha    & c_2 \cdot \alpha     & c_3 \cdot \alpha \\
	c_0 \cdot \alpha  & c_1                 & c_2 \cdot \alpha^2   & c_3 \cdot \alpha^2 \\
	c_0 \cdot \alpha  & c_1 \cdot \alpha^2  & c_2                  & c_3 \cdot \alpha^2 \\
	c_0 \cdot \alpha  & c_1 \cdot \alpha^2  & c_2 \cdot \alpha^2   & c_3
\end{bmatrix}.
\end{equation*}
The privacy and non-negativity constraints hold if $c_i \geq 0$ for every $i$. Imposing row-stochasticity, we can solve for the coefficients and get the unique solution: $c_1=c_2=c_3=1/(\alpha+1), c_0 = (1- 2 \alpha )/ (\alpha +1)$. Mechanism $U$ is only feasible if $c_0 \geq 0$, or equivalently $\alpha \leq 0.5$. 

Note that, so far, we used only three of the vertices adjacent to $v_0$. Suppose $v_0$ has $k>3$ neighbors. We actually can achieve stronger results by treating more consumers, each with a prior of uniform probability over only three vertices (one of which is $v_0$). Using the same arguments, and combining the partitions imposed by each of the consumers on $U$'s columns, we get that $U$'s positive columns are partitioned into columns of the forms: $\gamma_0 \cdot (1,\alpha, \alpha, \ldots, \alpha)^T$, $\gamma_1 \cdot (\alpha, 1, \alpha^2, \ldots, \alpha^2)^T$, $\gamma_2 \cdot (\alpha, \alpha^2, 1, \ldots, \alpha^2)^T$, \ldots , $\gamma_k \cdot (\alpha, \alpha^2, \alpha^2, \ldots, 1)^T$. Thus, a universal mechanism for all these consumers must have the structure:
\begin{equation*}
U=
\begin{bmatrix}
	c_0               & c_1 \cdot \alpha    & c_2 \cdot \alpha     & \ldots   & c_k \cdot \alpha \\
	c_0 \cdot \alpha  & c_1                 & c_2 \cdot \alpha^2   & \ldots   & c_k \cdot \alpha^2 \\
	c_0 \cdot \alpha  & c_1 \cdot \alpha^2  & c_2                  & \ldots   & c_k \cdot \alpha^2 \\
	\vdots            & \vdots              & \vdots               & \ddots   & \vdots \\
	c_0 \cdot \alpha  & c_1 \cdot \alpha^2  & c_2 \cdot \alpha^2   & \ldots   & c_k
\end{bmatrix}.
\end{equation*}
Imposing row-stochasticity, we can solve for the coefficients and get the unique solution: $c_i = 1/(\alpha +1)$ for every $i>0$ and $c_0 = (1-(k-1)\alpha)/(\alpha+1)$. Mechanism $U$ is only feasible if $c_0 \geq 0$, or equivalently $\alpha \leq 1/(k-1)$.

\end{proof}

% bibliography

\appendix

\remove{
\_section{Sufficiency of Oblivious Universal Mechanisms}
\label{app:non_oblivious}

Recall that, following~\cite{GRS09,GS10}, we only consider information consumers whose a priori probability is over the true value of the query answer $f(D)$ (and not, e.g., over the database entries). While this restriction over the information consumers weakens the positive results of~\cite{GRS09,GS10}, it strengthens our negative results. In both \cite{GRS09} and \cite{GS10}, it is shown that oblivious mechanisms minimize the worst-case disutilities of such information consumers, and hence only oblivious optimal mechanisms are considered. Furthermore, the universally optimal mechanism in~\cite{GRS09,GS10} is the geometric mechanism, which is oblivious.

In our negative results we only consider {\em oblivious} optimal mechanisms. This choice that needs to be justified, as it may be that it is possible to construct a {\em non-oblivious} universally optimal mechanism $U'$ from which all the (oblivious) optimal mechanisms are derivable via remapping. We now show that if this is the case, then it is possible to use $U'$ to construct an oblivious universally optimal mechanism $U$.

Writing $U'$ in matrix form, we let $u'_{D,k}$ be the probability that database $D$ is mapped by $U'$ to outcome $k$. Let 
$$E(i) = \{D\in \D^n: f(D) = i\},$$ 
i.e., $E(i)$ is the set of all the databases whose query result is $i$. Define a new mechanism $U=(u_{i,k})$ by averaging the probabilities of mechanism $U'$:
\begin{equation*}
	u_{i,k} = \frac{1}{\lvert E(i) \rvert} \sum_{D \in E(i)} u'_{D,k}.
\end{equation*}
Note that $U$ is oblivious, and following arguments from~ \cite{GRS09,GS10}, as $U'$ maintains differential privacy, so does $U$. We will show that $U$ is universally optimal.

Consider now a remapping $T=(t_{k,r})$ that may be applied by an information consumer to obtain her oblivious optimal mechanism, i.e., $T \circ U'$ is oblivious. We claim that $T\circ U$ behaves the same. As $T\circ U'$ is oblivious, we have that  for every $i \in \R_f$, for every output $r$, and for every $D_1,D_2 \in E(i)$: 
$$(T \circ U')_{D_1,r} =(T \circ U')_{D_2,r}.$$ 
Hence, for every $D_1,D_2 \in E(i)$ we have:
\begin{equation} \label{eq:U'T_obliviousness}
	\sum_k u'_{D_1,k} \cdot t_{k,r} = \sum_k u'_{D_2,k} \cdot t_{k,r}.
\end{equation}
Now, for every query result $i$ and $D_1 \in E(i)$:
\begin{eqnarray*}
(T \circ U')_{D_1,r} & = & \sum_k u'_{D_1,k} \cdot t_{k,r} \\
& = & \frac{1}{\lvert E(i) \rvert} \sum_{D \in E(i)}  \sum_k u'_{D,k} \cdot t_{k,r}  \quad\mbox{(by Equation~(\ref{eq:U'T_obliviousness}))} \\
&= &\sum_k \frac{1}{\lvert E(i) \rvert} \sum_{D \in E(i)} u'_{D,k} \cdot t_{k,r} \\
 & = &\sum_k u_{i,k} \cdot t_{k,r} = (T \circ U)_{i,r}.
\end{eqnarray*}

We get that for all the remappings $T$ that when applied on the outcome of $U'$ result with an oblivious mechanism, the derived mechanisms $T\circ U'$ and $T \circ U$ are identical. Hence, if an oblivious universally optimal mechanism does not exist for $f$, neither does a non-oblivious universally optimal mechanism.

}

\remove{
\_section{The Basic Impossibility Result With Continuous Range Mechanisms}
\label{app:continuous_range}

We now generalize the proof of Claim \ref{clm:impossibility_sumquery} to the case in which the universal mechanism $U$ might have a continuous range in $\mathbb{R}$.
\begin{myclaim}
Let $$X = \frac{1}{1+2\alpha} \cdot 
\begin{smallbmatrix}
	1 & \alpha & \alpha  \\
	\alpha & 1 & \alpha  \\
	\alpha & \alpha & 1
\end{smallbmatrix}; \; Y = \frac{1}{1+\alpha} \cdot 
\begin{smallbmatrix}
	1 & \alpha & 0 \\
	\alpha & 1 & 0 \\
	q & 1+\alpha-q & 0       
\end{smallbmatrix},$$
where $q \in \left[\alpha, 1\right]$. No $\alpha$-differentially private mechanism can derive both $X$ and $Y$.
\end{myclaim}
\begin{proof}
Assume for a contradiction that such a mechanism $U$ exists. As $X$ can be derived from $U$, there exists a remapping $T$ such that $X = T \circ U$. Similarly, there exists a remapping $W$ and an instance of $Y$ such that $Y = W \circ U$. For every $ i \in \{0,1,2\}$, denote by $f_i(x)$ the probability density function of the distribution of $U$'s output once the input is $i$. For $k \in \{0,1,2\}$ denote by $t_k(x)$ the probability of mechanism $T$ to remap value $x$ into $k$. Denote by $t^*(x)$ the probability of mechanism $T$ to remap value $x$ into a value other than $\{0,1,2\}$. Similarly define $w_0(x), w_1(x)$ and $w^*(x)$ for mechanism $W$. 
Note that for every $x \in \mathbb{R}$
\begin{eqnarray*}
	t_0(x) + t_1(x) + t_2(x) + t^*(x) & = & 1,\; \mbox{and} \\
	w_0(x) + w_1(x) + w^*(x) & = & 1.
\end{eqnarray*}
The probability of the composed mechanism $T \circ U$ to remap value $i$ into $k$ is given by
\begin{equation*}
	x_{i,k} = Pr\left[ T \circ U (i) = k \right] = \int_{\mathbb{R}} f_i(x) \cdot t_k(x) \, dx.
\end{equation*}
Note that the integral is well defined, as $f_i \cdot t_k \geq 0$ and $t_k(x) \leq 1$ for all $x$.
Similarly, 
\begin{equation*}
	y_{i,k} = Pr\left[ W \circ U (i) = k \right] = \int_{\mathbb{R}} f_i(x) \cdot w_k(x) \, dx.
\end{equation*}

Note that, unlike in the discrete case, there may be values $x$ such that $\alpha \cdot f_j(x) > f_i(x)$. But according to the differential privacy constraints, for every measurable subset $S \subseteq \mathbb{R}$ 
\begin{equation*}
	0 \leq \int_S f_i(x) \, dx - \alpha \int_S f_j(x) \, dx = \int_S (f_i(x) - \alpha f_j(x)) \, dx.
\end{equation*}
And so, $f_i(x) - \alpha f_j(x) < 0$ only on a zero-measure subset $S$. Therefore, 
\begin{equation}
f_i(x) \geq \alpha \cdot f_j(x) \quad \mbox{almost everywhere}. 
\label{eq:diffprivconstraints3}
\end{equation}

We now define the sets $A,B,C,D$ in correspondence to the values in $U$'s range which are remapped to the $0^{\text{th}}$, $1^{\text{st}}$, and $2^{\text{nd}}$ columns of X, similarly to the discrete range case:
\begin{align*}
	D & = \{ x \in \mathbb{R} : f_0(x) = f_1(x) = f_2(x) = 0 \},  \\
	A & = \{ x \in \mathbb{R} : t_0(x) > 0 \} \setminus D,      \\
	B & = \{ x \in \mathbb{R} : t_1(x) > 0 \} \setminus D,      \\
	C & = \{ x \in \mathbb{R} : t_2(x) > 0 \} \setminus D. 
\end{align*}

We now argue that $A,B,C,D$ form a partition of $\mathbb{R}$, except for maybe the negligible sets $A \cap B$, $A \cap C$, $B \cap C$ and $\mathbb{R} \setminus (A \cup B \cup C \cup D)$. By definition, $D$ is disjoint from $A,B,C$. In the following we show that $A,B,C$ are almost disjoint, and that $A \cup B \cup C \cup D$ almost equals $\mathbb{R}$, hence almost a partition.

In the case of the $0$ column of our specific mechanism $X$, $\alpha x_{0,0} = x_{1,0}$ and $\alpha x_{0,0} = x_{2,0}$, and we have for $j \neq 0$
\begin{align*}
	0 & = x_{j,0} - \alpha x_{0,0} = \int_\mathbb{R} f_j(x) \cdot t_0(x) \, dx - \alpha \int_\mathbb{R} f_0(x) \cdot t_0(x) \, dx \\ & = \int_{\mathbb{R}} (f_j(x)-\alpha f_0(x)) \cdot t_0(x) \, dx = \int_A (f_j(x)-\alpha f_0(x)) \cdot t_0(x) \, dx. 
\end{align*}
Considering the constraints in Equation~(\ref{eq:diffprivconstraints3}), the integrand is a non-negative function almost everywhere. We conclude these constraints are tight for all $x \in A$ except for a negligible subset. In other words, $\alpha f_0(x) = f_1(x) = f_2(x)$, almost everywhere in $A$. 
A similar analysis yields that $f_0(x) = \alpha f_1(x) = f_2(x)$ almost everywhere in $B$, and $f_0(x) = f_1(x) = \alpha f_2(x)$ almost everywhere in $C$. Since $f_0(x) > 0$ almost everywhere in $A$ (resp. $B$, $C$), we conclude that the sets $A,B$ and $C$ are almost disjoint. I.e., the tight constraints on almost all $x \in A$ exclude $x$ from $B$ and $C$, except maybe some negligible subsets of values in $\mathbb{R}$ which might appear in two (or three) of these sets.

Last, we show that the union of $A,B,C,D$ almost equals $\mathbb{R}$. Since the support of $X$ is $\{0,1,2\}$ and the support of $Y$ is $\{0,1\}$ we conclude that for every $i \in \{0,1,2\}$
\begin{equation*}
	\int_\mathbb{R} f_i(x) \cdot w^*(x) \, dx = \int_\mathbb{R} f_i(x) \cdot t^*(x) \, dx = 0 
\end{equation*}
Therefore, $t^*(x) = 0$ and $w^*(x) = 0 $, and hence $t_0(x) + t_1(x) + t_2(x) = 1 $ and $ w_0(x) + w_1(x) = 1 $, for every $x \in \mathbb{R}$ except for a subset negligible according to the measure $f_i$. Every such $x$ is in $ A \cup B \cup C $. For every other $x \in \mathbb{R}$, $f_0(x) = f_1(x) = f_2(x) = 0$ except for negligible subsets in $\mathbb{R}$, and so $ x \in A \cup B \cup C \cup D $ almost everywhere in $\mathbb{R}$.

Therefore, there exists a partition $A',B',C',D'$ such that the integrals over $A$ (resp. $B$, $C$ or $D$) equal the integrals over $A'$ (resp. $B'$, $C'$ or $D'$). We assume $A,B,C,D$ is indeed a partition.

Now observe the remapping $W$ and mechanism $Y$. Note that some privacy constraints in $Y$ are tight (e.g., $\alpha y_{0,0}=y_{1,0}$). Hence,
\begin{align*}
	0 &= y_{1,0} - \alpha y_{0,0} \\
	  &= \int_{\mathbb{R}} f_1(x) \cdot w_0(x) \, dx - \alpha \int_{\mathbb{R}} f_0(x) \cdot w_0(x) \, dx \\
	  &= \int_{\mathbb{R}} \left( f_1(x) - \alpha f_0(x) \right) \cdot w_0(x) \, dx \\
	  &= \int_{A} \left( f_1(x) - \alpha f_0(x) \right) \cdot w_0(x)\,  dx + \int_{B} \left( f_1(x) - \alpha f_0(x) \right) \cdot w_0(x) \, dx + \int_{C} \left( f_1(x) - \alpha f_0(x) \right) \cdot w_0(x) \, dx \\
	  &= \int_{B} \left( f_1(x) - \alpha f_0(x) \right) \cdot w_0(x) \, dx + \int_{C} \left( f_1(x) - \alpha f_0(x) \right) \cdot w_0(x) \, dx
\end{align*}
The fourth equality holds because $A,B,C,D$ is a partition of $\mathbb{R}$ and because for every $x\in D: \, f_0(x)=f_1(x)=0$. The last equality holds because $f_1(x) = \alpha f_0(x)$ almost everywhere in $A$. Notice that both integrals in the last addition are over non-negative values almost everywhere. Therefore, they are both non-negative. Since their sum is $0$, we get that 
\begin{align*}
	0 = \int_{C} \left( f_1(x) - \alpha f_0(x) \right) \cdot w_0(x) \, dx = (1-\alpha) \int_{C} f_0(x)\cdot w_0(x) \, dx 
\end{align*}
The equality holds because $f_0(x) = f_1(x)$ almost everywhere in $C$.

Similarly, using the tight constraint $\alpha y_{1,1}=y_{0,1}$, we get
\begin{equation*}
	0 = (1-\alpha) \int_{C} f_0(x)\cdot w_1(x) \, dx
\end{equation*}
Summing the last two equations together, we get
\begin{equation*}
	0 = (1-\alpha) \int_{C} f_0(x)\cdot (w_0(x) + w_1(x)) \, dx = (1-\alpha) \int_{C} f_0(x) \, dx ,
\end{equation*}
and we conclude
\begin{equation}
	\int_{C} f_0(x) \, dx = \int_{C} f_1(x) \, dx = \int_{C} f_2(x) \, dx = 0,
\label{eq:cIsInsignificant}
\end{equation}

This will be sufficient for us to show that the induced mechanism $Y$ is infeasible in the same way shown for the discrete case. We will estimate the entries of the last row in $Y$ and will show they sum up to a value less than $1$. Calculate $ y_{2,0} - y_{1,0}$ :
\begin{align*}
	y_{2,0} - y_{1,0} &= \int_{\mathbb{R}} (f_2(x)-f_1(x))w_0(x) \,dx \\
	& = \int_A (f_2(x)-f_1(x))w_0(x) \,dx + \int_B (f_2(x)-f_1(x))w_0(x) \,dx + \int_C (f_2(x)-f_1(x))w_0(x) \,dx \\
	& = \int_B (\alpha - 1) f_1(x) w_0(x) \, dx + \int_C (1-\alpha) f_2(x) w_0(x) \, dx \\
	& = (\alpha - 1) \int_B  f_1(x) w_0(x) \, dx \\
	& \leq 0
\end{align*}
The second equality holds because $A,B,C,D$ is a partition of $\mathbb{R}$ and because $f_1(x)=f_2(x)=0$ everywhere in $D$. The third equality holds because $f_2(x)=f_1(x)$ almost everywhere in $A$, \mbox{$f_2(x) =\alpha f_1(x)$} almost everywhere in $B$ and $f_1(x) = \alpha f_2(x)$ almost everywhere in $C$. The last equality holds because of Equation~(\ref{eq:cIsInsignificant}). The inequality holds because $\alpha < 1$ and $f_1(x),w_0(x) \geq 0$.

We get that $y_{2,0} \leq y_{1,0} = \alpha / (1+\alpha)$. Similarly, by calculating $y_{2,1} - y_{0,1}$ we get that $y_{2,1} \leq y_{0,1} = \alpha / (1+\alpha)$.
% the formula is cut at end of line. \mbox command cannot help here
Since $\alpha < 1$, these two entries in the last row of $Y$ do not sum up to 1, a contradiction to the mechanism $Y$ being row-stochastic.

\end{proof}
}

\remove{
\_section{Generalizing the Impossibility Result of Universally Optimal Mechanisms for Sum Queries}
\label{app:sum_generalizations}
\_subsection{Generalizing the Impossibility Result to $m > 2$} \_label{sec:m_gt_2}

Consider the case where the database consists of one record, and the possible values in this record are $0$ to $m$. Let 
$$X = \frac{1}{1+m \alpha} \cdot 
	\begin{smallbmatrix}
		1 & \alpha & \alpha & \cdots & \alpha \\
		\alpha & 1 &\alpha  & \cdots & \alpha \\
		\alpha & \alpha & 1 & \cdots & \alpha \\
		\vdots & \vdots & & \ddots & \vdots \\
		\alpha & \alpha &\alpha &\cdots &1
	\end{smallbmatrix};\;
	Y = \frac{1}{1+(m-1) \alpha} \cdot 
	\begin{smallbmatrix}
		1 & \alpha & \alpha & \cdots & \alpha & 0\\
		\alpha & 1 &\alpha  & \cdots & \alpha & 0\\
		\alpha & \alpha & 1 & \cdots & \alpha & 0\\
		\vdots & \vdots & & \ddots & \vdots \\
		\alpha & \alpha &\alpha &\cdots & 1 & 0\\
		q_1 & q_2 & q_3 & \cdots & q_{m} & 0
	\end{smallbmatrix},$$
where $\alpha \leq q_i \leq 1$ and $\sum_{i=1}^m q_i =1+(m-1)\alpha$.
Similar arguments to those used for the case $m=2$ show that $X$ is the unique optimal mechanism for an information consumer with loss function $\ell_{bin}$ and prior $p_0=p_1=\cdots = p_m = \frac{1}{m+1}$ in the Bayesian utility model and for an information consumer with support $S=\{0,1,\ldots ,m\}$ in the risk-averse utility model. Also, mechanisms of the form $Y$ are the only optimal mechanisms for the information consumers with loss function $\ell_{bin}$ and prior $p_0=p_1=\cdots = p_{m-1} = \frac{1}{m},\, p_m=0$ in the Bayesian model, and for an information consumer with support $S=\{0,1,\ldots ,m-1\}$ in the risk-averse model. Using similar arguments to those in the proof of Claim \ref{clm:impossibility_sumquery}, it follows that $X$ and $Y$ are not derivable from one single mechanism.

\_subsection{Generalizing the Impossibility Result to $n>1$}
\label{sec:n_gt_1}

Now consider the case where the number of records in the database is larger than $1$. We first deny the possibility of an oblivious universally optimal mechanism. Consider two consumers with loss function $\ell_{bin}$. The first consumer believes that the result of the sum query is bounded by $m$ (in the Bayesian case, the consumer holds a uniform prior over $\{0,\ldots,m\}$). No optimal mechanism for this consumer returns values larger than $m$, so in the mechanism matrix the columns from $m+1$ on  contain zeros. Refer to some optimal mechanism for this consumer as $X'$. Ignoring rows and columns of $X'$ that correspond to values greater than $m$,  the remaining entries exactly form the mechanism $X$ of the previous subsection. (Observe that such an extension of mechanism $X$ is indeed feasible, as any row which pertains to a value greater than $m$ can be identical to the row which pertains to the value $m$, and so the privacy constraints hold. Such a mechanism is also optimal, as the utility is a function of only the rows $\{0,1,\ldots,m\}$, due to the consumer's prior, so we cannot achieve a better utility than the utility gained by mechanism $X$). The second consumer believes that the query result cannot be larger than $m-1$  (in the Bayesian case, the consumer holds a uniform prior over $\{0,\ldots,m-1\}$). Refer to some optimal mechanism for this consumer as $Y'$. A similar argument shows, that ignoring rows and columns that pertains to values greater than $m$, the remaining entries match the mechanism $Y$ of the previous subsection. 
Assume for a contradiction that $X'$ and $Y'$ are both derivable from some mechanism $U'$. Therefore there exist remappings $T,S$ such that $X' = U'T$ and $Y'=U'S$. Let $U$ be the mechanism $U'$ reduced to only the inputs $\{0,1,\ldots,m\}$. Reducing $U'$ to $U$, we get that $X=UT$ and $Y=US$. According to the previous subsection these two mechanisms cannot be derived from a single oblivious mechanism, due to the same arguments in the proof of Claim \ref{clm:impossibility_sumquery}. Thus, a contradiction.

Now suppose for a contradiction that both the mechanisms are derived from a single non-oblivious mechanism $U^*$. This means that $U^*$'s input space is different databases rather than query results. Suppose there is a remapping $T$ such that $X^*=U^* T$. This means that the indexes of $X^*$'s rows should be databases as well. We assume $X^*$ is oblivious (if we let the consumers choose non-oblivious optimal mechanisms then there is no chance of universally optimal mechanisms even for count queries \cite{GRS09}). Therefore, applying $U^*$ on two databases with the same query result and then applying T on $U^*$'s output,
yields identical rows in $X^*$ (which is described as a single row in the oblivious matrix $X$ above). Note that although $X^*$'s input and output spaces are discrete (and so we can refer to $X^*$ as a matrix), we assume nothing on $U^*$'s outputs and $T$'s inputs.
Reducing $U^*$ to an input space of only $m+1$ databases with different query results and applying the remapping $T$ on this reduced mechanism's output, yields mechanism $X$ completely. 
Similarly, applying some remapping $S$ on the same reduced mechanism yields mechanism $Y$. 
Now reduce $U^*$ to inputs which are the databases $(0,0,\ldots,0,q)$ where $q$ is any possible record value. Refer to this mechanism as $U$. According to the assumptions, we get that $X=UT,Y=US$. Also note that every two possible inputs of $U$ are neighboring databases, and so $U$ must satisfy privacy constraints as any oblivious mechanism. Therefore, we get a simple reduction to the case of an oblivious mechanism $U$, and the same impossibility result applies also to the case of non-oblivious universal mechanism. Thus, we conclude Theorem \ref{thm:sum_queries}

}

\remove{\_section{Impossibility of Universally Optimal Mechanisms for Other Query Types}
\label{app:histograms}

\remove{
\_subsection{Histogram Queries}
We can relate to count queries as a histogram of a partition of the records into two categories: those which satisfy the predicate, and those which do not. The exact result of a count query is actually the result of such a histogram. Consider a query for a histogram on the database records which partitions them into three categories or more. We now discuss the details of the proof for Theorem \ref{thm:histograms}.

Once again, consider first the case where there is only one record in the database, and the query is for a histogram which partitions the possible records into three categories. The only possible results for such a query are $\left(1,0,0\right)$, $\left(0,1,0\right)$ and $\left(0,0,1\right)$. Notice that all these histograms result from neighboring databases. Now consider information consumers who are described with the $\ell_{bin}$ loss function or with the $\ell_1$ loss function. (In the case of a histogram over one record, they are the same, resulting with $0$ if the exact result and the output are identical and with a constant if not). Refer to the first possible result as $0$, the second possible result as $1$, and the third possible result as $2$. Notice now that we have exactly the same constraints for valid mechanisms as we had for the sum query with just one record. Also, the utility expression for each of the consumers is the same. The problem of universally optimizing the utility for all $\ell_{bin}$ information consumers (or $\ell_1$ consumers) is now reduced to the same problem for sum queries. According to Section \ref{sec:sum_queries}, universally optimizing the utility for all such consumers is impossible, and so it is impossible to construct a mechanism for this specific case as well.

We now generalize this result for histograms over larger databases and partitions of any number of categories larger than $2$. First, consider the case of querying one record for a histogram of $c \geq 3$ number of categories. This can easily be reduced to problems we have already answered negatively. One way is to notice that as in the case where $c=3$ (in which we reduced this problem to the problem of sum queries where the records' values bound is $m=2$), larger values of $c$ can easily be reduced to sum queries with larger bounds on the records' values $m$. For every number of $c$ partitions, there are exactly $c$ possible results for the histogram over one record. They are all the results of neighboring databases. Refer to these results as $0,1,2,\ldots,c-1$. Again, this is exactly like constructing a universally optimal mechanism for sum queries over one record, in which the bound on its values is $m=c-1$. This is impossible as was shown in Appendix \ref{app:sum_generalizations}. Another way to be convinced is to refer to the partitions as $A_1, A_2, \ldots, A_c$. Now consider only consumers whose utility functions are calculated by the number of records in $A1, A2, A_3 \cup A_4 \cup \ldots \cup A_c$. These are still monotone loss functions. This reduces the current problem to the problem of histograms over a partition of only $3$ categories, to which we already proved negative results.

We now generalize this result further to any size of databases. The same argument that was applied in Appendix \ref{app:sum_generalizations} for sum queries, is applied here as well (even for the case of non-oblivious universal mechanisms). Consider only consumers with a prior such that all records except perhaps one fall into one specific category of the histogram. Querying for a histogram on such a database reduces to the result of the same histogram over one record only. Even if we consider only these consumers, we know that no one single mechanism can optimize their utilities over all possible mechanisms. Since there is no such mechanism that optimizes these consumers' utilities, there is obviously no mechanism that yields optimized results for all possible consumers. Therefore, even for larger databases, there is no universally optimal private mechanism for histogram queries.
}

\remove{
\_subsection{Bundles of Count Queries}
We now pay attention to another generalization of single count queries. A possible query might be a bundle of simple (non trivial) count queries to which the mechanism answers simultaneously, as they were just one query. We discuss the impossibility of a universally optimal mechanism for such bundles, as stated in Theorem \ref{thm:bundles}.

Note that the loss of each information consumer is a function on some query results and outputs at once. A consumer's disutility is not necessarily a sum of the losses for all separate basic queries. Otherwise the mechanism uses the universally optimal geometric mechanism for count queries, and replies each query independently and separately. In our scenario, consumers might have auxiliary knowledge regarding some query results at once, and the loss function that describes them is also a function of some outputs at once. For instance, consider an $\ell_{bin}$ consumer and two simultaneous non-trivial count queries. Such a consumer pays no loss if both the results he uses are correct, but he will pay a single constant loss if any of the two results he uses is incorrect, and also if both of them are incorrect. This is different from the case of a consumer who queries two different mechanisms for two different count queries. 

Such a generalization of count queries proves to be no different than the other intuitive generalizations we have already discussed. Note that two simultaneous non-trivial different predicates actually partition the records domain into $4$ categories: those which satisfy both predicates, those which satisfy none of them, and those which satisfy just the first or just the second. If one of the predicates is a subset of the other, then the predicates partition the domain into only $3$ categories. In each case, the reply for such two count queries is the same as the reply for a histogram query with $3$ or $4$ categories. We have already shown that no universally optimal mechanism exists for histogram queries (which are not identical to a single count query). Therefore, no universally optimal mechanism exists for a bundle of two count queries, as well. The same argument shows that no universally optimal mechanism exists for a bundle of any number (larger than 1) of different non-trivial count queries.
}
}

\end{document}